\documentclass[orivec]{llncs}

\usepackage{amssymb}
\usepackage{graphicx}
\usepackage{listings}
\usepackage{arydshln}
\usepackage{amsmath}
\usepackage{color}

\usepackage{amsfonts}
\usepackage{xspace}
\usepackage{latexsym}
\usepackage{wasysym} %
\usepackage{stmaryrd}
\usepackage{alltt}
\usepackage{mathpartir}
\usepackage[ligature,reserved,inference]{semantic}
\usepackage[lined]{algorithm2e}
\usepackage{at}
\usepackage{alltt}
\usepackage{listings}
\usepackage[lined]{algorithm2e}
\usepackage{wrapfig}

\newcommand{\atCP}{\makeatletter @\makeatother}

\usepackage{pslatex}

\usepackage[numbers]{natbib}

\newcommand{\angles}[1]{\ensuremath{\langle {#1} \rangle}}

\renewcommand{\u}{\cup}

\let\oldae=\ae
\renewcommand{\ae}{\oldae\xspace}

\newcommand{\set}[1]{{\{ #1 \}}}

\newcommand{\tup}[1]{\angles{#1}}

\mathlig{,.,}{,\ldots,}
\mathlig{,..}{,\ldots}
\mathlig{..,}{\ldots,}
\mathlig{-->}{\longrightarrow}
\mathlig{<=>}{\Leftrightarrow}
\mathlig{==>}{\implies}
\mathlig{|->}{\mapsto}
\mathlig{<-}{\gets}
\mathlig{->}{\mathbin{\rightarrow}}
\mathlig{~>}{\mathbin{\rightharpoonup}}
\mathlig{..}{\ldots}
\mathlig{/|}{\land}
\mathlig{|/}{\lor}
\mathlig{=>}{\mathbin{\Rightarrow}}
\mathlig{<=}{\leq}
\mathlig{==}{\equiv}
\mathlig{~=}{\approx}
\mathlig{!=}{\neq}
\mathlig{|-}{\vdash}
\mathlig{|=}{\models}
\mathlig{~}{\sim}

\newcommand{\asgreekstyle}{\mathrm} %
\newcommand{\greek}[1]{\ensuremath{#1}\xspace}
\newatcommand a {\greek{\alpha}}
\newatcommand b {\greek{\beta}}
\newatcommand c {\greek{\xi}}
\newatcommand C {\greek{\Xi}}
\newatcommand d {\greek{\delta}}
\newatcommand D {\greek{\Delta}}
\newatcommand e {\greek{\varepsilon}}
\newatcommand E {\greek{\asgreekstyle{E}}}
\newatcommand f {\greek{\phi}}
\newatcommand F {\greek{\Phi}}
\newatcommand g {\greek{\gamma}}
\newatcommand G {\greek{\Gamma}}
\newatcommand h {\greek{\varrho}}
\newatcommand j {\greek{\iota}}
\newatcommand k {\greek{\kappa}}
\newatcommand l {\greek{\ell}\xspace}
\newatcommand ll {\greek{\lambda}}
\newatcommand L {\greek{\Lambda}}
\newatcommand m {\greek{\mu}}
\newatcommand n {\greek{\eta}}
\newatcommand o {\greek{\theta}}
\newatcommand oo {\greek{\vartheta}}
\newatcommand O {\greek{\Theta}}
\newatcommand p {\greek{\varphi}}
\newatcommand P {\greek{\Pi}}
\newatcommand q {\greek{\epsilon}}
\newatcommand Q {\greek{\Phi}} %
\newatcommand r {\greek{\rho}}
\newatcommand s {\greek{\sigma}}
\newatcommand ss {\greek{\sigma}}
\newatcommand S {\greek{\Sigma}}
\newatcommand t {\greek{\tau}}
\newatcommand u {\greek{\upsilon}}
\newatcommand U {\greek{\Upsilon}}
\newatcommand v {\greek{\nu}}
\newatcommand w {\greek{\omega}}
\newatcommand W {\greek{\Omega}}
\newatcommand x {\greek{\chi}}
\newatcommand y {\greek{\psi}}
\newatcommand Y {\greek{\Psi}}
\newatcommand z {\greek{\zeta}}

\newatcommand 0 {\emptyset}
\newatcommand \ {\setminus}

\newcommand{\mathfnstyle}[1]{\ensuremath{\mathrm{#1}}}
\reservestyle{\mathfn}{\mathfnstyle}

\newcommand{\metalangkeywordstyle}[1]{\ensuremath{\mathsf{#1}}}
\reservestyle{\metalangkeyword}{\metalangkeywordstyle}

\newcommand{\semanticdomainstyle}[1]{%
  \ensuremath{\mathchoice%
    {\mbox{\normalfont\ensuremath{#1}}}%
    {\mbox{\normalfont\ensuremath{#1}}}%
    {\mbox{\normalfont\scriptsize\ensuremath{#1}}}%
    {\mbox{\normalfont\tiny\ensuremath{#1}}}}}

\reservestyle{\semanticdomain}{\semanticdomainstyle} 

\newcommand{\proglangkeywordstyle}[1]{\ensuremath{\mathbf{#1}}\xspace}
\reservestyle{\proglangkeyword}{\proglangkeywordstyle}

\mathfn{min,max,log,dom,range}

\semanticdomain{%
  Bools[\mathbb{B}],
  Nats[\mathbb{N}],
  Lab[\mathbb{A}],
  Ints[\mathbb{Z}],
  Reads[\mathbb{R}],
  Writes[\mathbb{W}],
  Events[\mathbb{E}],
  Confs[\mathbb{C}]
}

\semanticdomain{%
  Methods[\mathbb{M}],
  Vals[\mathbb{V}],
  Labels[\mathbb{L}],
  Ops[\mathbb{O}],
  Acts[\mathbb{A}],
}

\semanticdomain{%
  Vars, Ids[IDs],
  Frames, Stmts, Tasks, Configs, MsgQueues, Locks,
}
  
\lstdefinelanguage{bpl}{%
  keywords={%
    procedure,var,int,while,assume,assert,true,false,call,return,newLock,acquire,release,
  },
  morecomment=[l]{//},
}

\lstdefinelanguage{program}{%
  keywords={%
    var,const,bool,int,true,false,%
    while,do,if,then,else,assume,assert,proc,call,return,%
    task,async,yield,wait,%
  },
  morecomment=[l]{//},
  morecomment=[s]{/*}{*/},
  morecomment=[n]{(*}{*)},
  mathescape=true,
  escapeinside=`',
}

\proglangkeyword{%
  var, const, bool, int, true, false,
  while[while\ ], do[\ do\ ], if[if\ ], then[\ then\ ], else[\ else\ ],
  assume[assume\ ], return[return\ ], proc, call[call\ ], async, goto
  newThread, newLock, acquire, release, sendMsg[sendMsg\ ]
}

\metalangkeyword{true,false}
\metalangkeyword{ret,op}
\mathfn{past,future,width}
\metalangkeyword{push,pop}
\metalangkeyword{read,write,eventId,id,main}

\newcommand{\match}{\mathbin{\mapstochar\relbar\mapsfromchar}}

\lstnewenvironment{program}{%
  \lstset{%
    language={program},
    basicstyle=\scriptsize,
    columns=flexible,
    breaklines=true,
    tabsize=4,
    escapeinside={(*@}{@*)}
  }
}{}

\lstMakeShortInline"

\newtheorem{lem}{Lemma}

\title{Proving linearizability using forward simulations}
\author{Ahmed Bouajjani\inst{1} \and Michael Emmi\inst{2} \and Constantin Enea\inst{1} \and Suha Orhun Mutluergil\inst{3}}
\institute{IRIF, University Paris Diderot \& CNRS, \email{\{abou,cenea\}@irif.fr} \and Nokia Bell Labs \email{michael.emmi@nokia.com} \and Koc University\email{smutluergil@ku.edu.tr}} 
\begin{document}

\maketitle

\begin{abstract}
Linearizability is the standard correctness criterion concurrent data structures such as stacks and queues. %
It allows to establish observational refinement between a concurrent implementation and an atomic reference implementation.
Proving linearizability requires identifying linearization points for each method invocation along all possible computations, leading to valid sequential executions, or alternatively, establishing forward \emph{and} backward simulations. In both cases, carrying out proofs is hard and complex in general. In particular, backward reasoning is difficult in the context of programs with data structures, and strategies for identifying statically linearization points cannot be defined for all existing implementations.  In this paper, we show that, contrary to common belief, many such complex implementations, including, e.g., the Herlihy\&Wing Queue and the Time-Stamped Stack, can be proved correct using only forward simulation arguments. This leads to simple and natural correctness proofs for these implementations that are amenable to automation. 

\end{abstract}

\section{Introduction}
Programming efficient concurrent implementations of atomic collections, e.g., stacks and queues, is error prone. To minimize synchronization overhead between concurrent method invocations, implementors avoid blocking operations like lock acquisition, allowing methods to execute concurrently. However, concurrency risks unintended inter-operation interference, and risks conformance to atomic reference implementations. Conformance is formally captured by \emph{(observational) refinement}, which assures that all behaviors of programs
using these efficient implementations would also be possible were the atomic reference implementations used instead.

Observational refinement can be formalized as a trace inclusion problem, and the latter can itself be reduced to an invariant checking problem, but this requires in general introducing history and prophecy variables~\cite{DBLP:journals/tcs/AbadiL91}. Alternatively, verifying refinement requires in general establishing a forward simulation {\em and} a backward simulation~\cite{DBLP:journals/iandc/LynchV95}. While simulations are natural concepts, backward reasoning, corresponding to the use of prophecy variables, is in general hard and complex for programs manipulating data structures.
Therefore, a crucial issue is to understand the limits of forward reasoning in establishing refinement. More precisely, an important question is to determine for which concurrent abstract data structures, and for which classes of implementations, it is possible to carry out a refinement proof using only forward simulations.

To get rid of backward simulations (or prophecy variables) while preserving completeness w.r.t. refinement, it is necessary to have reference implementations 
that are {\em deterministic}. Interestingly, determinism allows also to simplify the forward simulation checking problem. Indeed, in this case, this problem can be reduced to an invariant checking problem. Basically, the simulation relation can be seen as an invariant of the system composed of the two compared programs. Therefore, existing methods and tools for invariant checking can be leveraged in this context. %

But, in order to determine precisely what is meant by determinism, an important point is to fix the alphabet of observable events along computations. 
Typically, to reason about refinement between two library implementations, the only observable events are the calls and returns corresponding to the method invocations along computations. This means that only the external interface of the library is considered to compare behaviors, and nothing else from the implementations is exposed. Unfortunately, it can be shown that in this case, it is impossible to have deterministic atomic reference implementations for common data structures such as stacks and queues (see, e.g., \cite{DBLP:conf/cav/SchellhornWD12}). Then, an important question is what is the necessary amount of information that should be exposed by the implementations to overcome this problem ?

One approach addressing this question is based on linearizability \cite{journals/toplas/HerlihyW90} and its correspondence with refinement %
\cite{journals/tcs/FilipovicORY10,DBLP:conf/popl/BouajjaniEEH15}. Linearizability of a computation (of some implementation) means that each of the method invocations can be seen as happening at some point, called {\em linearization point}, occurring somewhere between the call and return events of that invocation. The obtained sequence of linearization points along the computation should define a sequence of operations that is possible in the atomic reference implementation. Proving the existence of such sequences of linearization points, for all the computations of a concurrent library, is a complex problem \cite{journals/iandc/AlurMP00,conf/esop/BouajjaniEEH13,DBLP:conf/netys/Hamza15}. 
However, proving linearizability becomes less complex when linearization points are fixed for each method, i.e., associated with the execution of a designated statement in its source code \cite{conf/esop/BouajjaniEEH13}. In this case, we can consider that libraries expose in addition to calls and returns, events signaling linearization points. By extending this way the alphabet of observable events, it becomes straightforward to define {\em deterministic} atomic reference implementations. Therefore, proving linearizability can be carried out using forward simulations when linearization points are fixed, e.g.,~\cite{conf/ppopp/VafeiadisHHS06,conf/cav/AmitRRSY07,conf/vmcai/Vafeiadis09,conf/tacas/AbdullaHHJR13}.
Unfortunately, this approach is not applicable to efficient implementations such as the LCRQ queue~\cite{DBLP:conf/ppopp/MorrisonA13} (based on the principle of the Herlihy\&Wing queue \cite{journals/toplas/HerlihyW90}), and the Time-Stamped Stack \cite{DBLP:conf/popl/DoddsHK15}. The proofs of linearizability of these implementations are highly nontrivial, very involved, and hard to read, understand and automatize. Therefore, the crucial question we address is what is precisely the kind of information that is necessary to expose in order to obtain deterministic atomic reference implementations for such data structures, allowing to derive simple and natural linearizability proofs for such complex implementations, based on forward simulations, that are amenable to automation ?

We observe that the main difficulty in reasoning about these implementations is that, linearization points of enqueue/push operations occurring along some given computation, depend in general on the linearization points of dequeue/pop operations that occur arbitrarily far in the future. Therefore, since linearization points for enqueue/push operations cannot be determined in advance, the information that could be fixed and exposed can concern only the dequeue/pop operations. 

One first idea is to consider that linearization points are fixed for dequeue/pop methods and only for these methods. We show that under the assumption that implementations expose linearizations points for these methods, it is possible to define deterministic atomic reference implementations for both queues and stacks. We show that this is indeed useful by providing a simple proof of the Herlihy\&Wing queue (based on establishing a forward simulation) that can be carried out as an invariant checking proof.

However, in the case of Time-Stamped Stack, fixing linearization points of pop operations is actually too restrictive. Nevertheless, we show that our approach can be generalized to handle this case. The key idea is to reason about what we call {\em commit points}, and  that correspond roughly speaking to the last point a method accesses to the shared data structure during its execution. We prove that by exposing commit points (instead of linearization points) for pop methods, we can still provide deterministic reference implementations. We show that using this approach leads to a quite simple proof of the Time-Stamped Stack, based on forward simulations.

\section{Preliminaries}
\vspace{-2mm}
We formalize several abstraction relations between libraries using a simple
yet universal model of computation, namely labeled transition systems (LTS).
This model captures shared-memory programs with an arbitrary number of threads,
abstracting away the details of any particular programming system irrelevant to
our development.

A \emph{labeled transition system} (LTS) $A=(Q,\Sigma, s_0, \delta)$ over the 
possibly-infinite alphabet $\Sigma$ is a possibly-infinite set $Q$ of states with
initial state $s_0 \in Q$, and a transition relation $\delta \subseteq Q \times @S \times
Q$. The $i$th symbol of a sequence $@t \in @S^*$ is denoted $@t_i$, and the empty
sequence is denoted by $\epsilon$.
An \emph{execution} of $A$ is an alternating sequence of states and transition labels (called also actions)
$\rho = s_0, e_0,s_1\ldots e_{k-1},s_k$ for some $k>0$ such that $\delta(s_i, e_i, s_{i+1})$
for each $i$ such that $0\leq i<k$. We write $s_i\xrightarrow{e_i\ldots e_{j-1}}_A s_j$ as shorthand for 
the subsequence $s_i,e_i,...,s_{j-1},e_{j-1},s_j$ of $\rho$, for any $0\leq i\leq j <k$
(in particular $s_i\xrightarrow{\epsilon}s_i$).
The projection $@t| \Gamma$ of a sequence $@t$ is the maximum subsequence of $@t$ over
 alphabet $\Gamma$. This notation is extended to sets of sequences as usual.
A \emph{trace} of $A$ is the projection $\rho | \Sigma$ of an execution $\rho$ of $A$. 
The set of executions, resp., traces, of an LTS $A$ is denoted by $E(A)$, resp., $Tr(A)$.
An LTS is \emph{deterministic} if for any state $s$ and any sequence $@t\in \Sigma^*$, there is at most
one state $s'$ such that $s\xrightarrow{@t}s'$. More generally, for an alphabet $\Gamma\subseteq \Sigma$,
an LTS is \emph{$\Gamma$-deterministic} if for any state s and any sequence $@t\in \Gamma^*$, there
is at most one state $s'$ such that $s\xrightarrow{@t'}s'$ and $@t$ is a subsequence of $@t'$.

\subsection{Libraries}
Programs interact with libraries by calling named library \emph{methods}, which
receive \emph{parameter values} and yield \emph{return values} upon completion.
We fix arbitrary sets $\<Methods>$ and $\<Vals>$ of method names and
parameter/return values. 
We fix an arbitrary set $\<Ops>$ of operation identifiers, and for given sets
$\<Methods>$ and $\<Vals>$ of methods and values, we fix the sets
\vspace{-2mm}
\begin{align*}
  & C = \set{ inv(m,d,k) : m \in \<Methods>, d \in \<Vals>, k \in \<Ops> }
  \text{ and } 
  R = \set{ ret(m,d,k) : m \in \<Methods>, d \in \<Vals>, k \in \<Ops> }  
\end{align*}

\vspace{-2mm}
\noindent
of \emph{call actions} and \emph{return actions}; each call action $inv(m,d,k)$
combines a method $m \in \<Methods>$ and value $d \in \<Vals>$ with an
\emph{operation identifier} $k \in \<Ops>$. Operation identifiers are used to
pair call and return actions. 
We may omit the second field from a call/return action $a$ for methods that have no inputs or return values.
For
notational convenience, we take $\<Ops>=\<Nats>$ for the rest of the paper.

A \emph{library} is an LTS over alphabet $\Sigma$ such that $C \u R\subseteq \Sigma$. 
We assume that the traces of a library satisfy standard well-formedness properties, 
e.g., return actions correspond to previous call actions, which for lack of 
space are delegated to Appendix~\ref{app:prelim}. An operation $k$ is called \emph{completed} in a trace $\tau$ when
$ret(m,d,k)$ occurs in $\tau$, for some $m$ and $d$. Otherwise, it is called \emph{pending}.

The projection of a library trace over $C\cup R$ is called a \emph{history}. The set of histories of a library $L$ is denoted by $H(L)$.
Since libraries only dictate methods’ executions between their respective calls and returns, for any history they admit, they must also 
admit histories with weaker inter-operation ordering, in which calls may happen earlier, and/or returns later. A 
history $h_1$ is \emph{weaker} than a history $h_2$, written $h_1 \sqsubseteq h_2$, 
if{f} there exists a history $h_1'$
obtained from $h_1$ by appending return actions, and deleting call actions,
s.t.:
  $h_2$ is a permutation of $h_1'$ that preserves the order between
  return and call actions, i.e.,~if a given return action occurs before a given
  call action in $h_1'$, then the same holds in $h_2$.
A library $L$ is called \emph{atomic} when there exists a set $S$ of sequential histories such that 
$H(L)$ contains every weakening of a history in $S$.
Atomic libraries are often considered as specifications for concurrent objects. 
Libraries can be made atomic by guarding their methods bodies with global lock acquisitions.

A library $L$ is called a \emph{queue implementation} when $\<Methods>=\{enq,deq\}$ ($enq$ is the method that enqueues a value and $deq$ is the method removing a value) and $\<Vals>=\<Nats>\cup\{{\tt EMPTY}\}$ where {\tt EMPTY} is the value returned by $deq$ when the queue is empty. Similarly, a library $L$ is called a \emph{stack implementation} when $\<Methods>=\{push,pop\}$ and  $\<Vals>=\<Nats>\cup\{{\tt EMPTY}\}$. For queue and stack implementations, we assume that the same value is never added twice, i.e., for every trace $@t$ of such a library and every two call actions $inv(m,d_1,k_1)$ and $inv(m,d_2,k_2)$ where $m\in \{enq,push\}$ we have that $d_1\neq d_2$. As shown in several works~\cite{conf/tacas/AbdullaHHJR13,DBLP:conf/icalp/BouajjaniEEH15}, this assumption is without loss of generality for libraries that are data independent, i.e., their behaviors are not influenced by the values added to the collection, which is always the case in practice. On a technical note, this assumption is used to define ($\Gamma$-)deterministic abstract implementations of stacks and queues in Section~\ref{sec:queues} and Section~\ref{sec:stacks}.

\subsection{Refinement and Linearizability}
Conformance of a library $L_1$ to a specification given as an ``abstract'' library $L_2$ 
is formally captured by \emph{(observational) refinement}. Informally, 
we
say $L_1$ refines $L_2$ if{f} every computation of every program
using $L_1$ would also be possible were $L_2$ used instead. We assume that a program can 
interact with the library only through call and return actions, and thus refinement can be defined
as history set inclusion. Refinement is equivalent to the \emph{linearizability} criterion~\cite{journals/toplas/HerlihyW90} 
when $L_2$ is an atomic library~\cite{journals/tcs/FilipovicORY10,DBLP:conf/popl/BouajjaniEEH15}.

\begin{definition}
A library $L_1$ \emph{refines} another library $L_2$ if{f} $H(L_1) \subseteq H(L_2)$.
\end{definition}

Linearizability~\cite{journals/toplas/HerlihyW90} requires that every history of a concurrent library $L_1$ can be 
``linearized'' to a sequential history admitted by a library $L_2$ used as a specification. 
Formally, a sequential history $h_2$ with only complete operations is called a \emph{linearization} of a history $h_1$ when $h_1 \sqsubseteq h_2$.
A history $h_1$ is \emph{linearizable} w.r.t.~a library $L_2$ if{f} there exists a linearization $h_2$ of $h_1$ such that 
$h_2 \in H(L_2)$. A library $L_1$
is \emph{linearizable} w.r.t. $L_2$, written $L_1 \sqsubseteq L_2$, if{f}
each history $h_1 \in H(L_1)$ is linearizable w.r.t. $L_2$.

\vspace{-1.5mm}
\begin{theorem}[\cite{journals/tcs/FilipovicORY10,DBLP:conf/popl/BouajjaniEEH15}]
  $L_1 \sqsubseteq L_2$ if{f} $L_1$ refines $L_2$,
  if $L_2$ is atomic.
\vspace{-1.5mm}
\end{theorem}

In the rest of the paper, we discuss methods for proving refinement (and thus, linearizability) focusing mainly on queue and stack implementations.

\vspace{-3.5mm}
\section{Refinement Proofs}
\vspace{-1.5mm}
Library refinement is the instance of a more general notion of refinement between LTSs
which for some alphabet $\Gamma$ of \emph{observable actions} is defined as the inclusion of sets of 
traces projected on $\Gamma$. Library refinement corresponds to the case $\Gamma=C\cup R$. 
Typically, $\Gamma$-refinement between two LTSs $A$ and $B$ is proved using \emph{simulation relations} which roughly, require that 
$B$ can mimic every step of $A$ using a (possibly empty) sequence of steps. Mainly, there are two kinds of simulation
relations, forward or backward, depending on whether the preservation of steps is proved starting from a similar state
forward or backward. It has been shown
that $\Gamma$-refinement is equivalent to the existence of \emph{backward simulations}, modulo the addition of history variables
that record events in the implementation, and to the existence of \emph{forward simulations} provided that the right-hand side
LTS $B$ is $\Gamma$-deterministic~\cite{DBLP:journals/tcs/AbadiL91,DBLP:journals/iandc/LynchV95}. 
We focus on proofs based on forward simulations because they are easier to automatize.

In general, forward simulations are \emph{not} a complete proof method for library refinement because libraries are not 
$C\cup R$-deterministic (the same sequence of call/return actions can lead to different states depending on the interleaving of the internal actions).
However, there are classes of atomic libraries, e.g., libraries with ``fixed linearization points'' (defined later in this section), 
for which it is possible to identify a larger alphabet $\Gamma$ of observable actions (including call/return actions), 
and implementations that are $\Gamma$-deterministic. For queues and stacks, 
Section~\ref{sec:queues} and Section~\ref{sec:stacks} define other such classes of implementations that cover
all the implementations that we are aware of.

Let $L_1=(Q_1,\Sigma, s_0^1, \delta_1)$ and $L_2=(Q_2,\Sigma, s_0^2, \delta_2)$ be two libraries over $\Sigma_1$ and $\Sigma_2$, resp., such that $C\cup R \subseteq \Sigma_1\cap\Sigma_2$. Also, let $\Gamma$ be a set of actions s.t. $C\cup R\subseteq \Gamma\subseteq \Sigma_1\cap\Sigma_2$.

\begin{definition}\label{def:gref}
The library $L_1$ \emph{$\Gamma$-refines} $L_2$ if{f} $Tr(L_1) | \Gamma \subseteq Tr(L_2) | \Gamma$.
\end{definition}

Notice that $\Gamma$-refinement implies refinement for any $\Gamma$ as in Definition~\ref{def:gref}.

We define a notion of \emph{forward} simulation that can be used to prove $\Gamma$-refinement (a dual notion of \emph{backward} simulation is defined in 
Appendix~\ref{app:backSim}). 
For a relation $R\subseteq A\times B$, $R[X]$ is the set of elements related by $R$ to elements of $X$, i.e., $R[X]=\set{y:\exists x\in X.\ R(x,y)}$.

\vspace{-1.5mm}
\begin{definition}
A relation $\mathit{fs} \subseteq Q_{1} \times Q_{2}$ is called a \emph{$\Gamma$-forward simulation} from $L_1$ to $L_2$ if{f} $\mathit{fs}[s_0^1] = \{s_0^2 \}$ and:
\vspace{-1.5mm}
\begin{itemize}
\item If $(s,\gamma,s') \in \delta_1$, for some $\gamma\in \Gamma$, and $u \in \mathit{fs}[s]$, then there exists $u' \in \mathit{fs}[s']$ such that $u \xrightarrow{@s} u'$, $@s_i=\gamma$, for some $i$, and $@s_j\in \Sigma_2\setminus\Gamma$, for each $j\neq i$.
\item If $(s,e,s') \in \delta_1$, for some $e \in \Sigma_1\setminus \Gamma$ and $u \in \mathit{fs}[s]$, then there exists $u' \in \mathit{fs}[s']$ such that $u \xrightarrow{\sigma} u'$ and $\sigma\in (\Sigma_2\setminus\Gamma)^*$.  
\end{itemize}
\vspace{-3.5mm}
\end{definition}
A $\Gamma$-forward simulation requires that every step of $L_1$ corresponds to a sequence of steps of $L_2$. To imply $\Gamma$-refinement, every step of $L_1$ labeled by an observable action $\gamma\in \Gamma$ should be simulated by a sequence of steps of $L_2$ where exactly one transition is labeled by $\gamma$ and all the other transitions are labeled by non-observable actions.

The following shows the soundness and the completeness of $\Gamma$-forward simulations (when $L_2$ is $\Gamma$-deterministic). It is an instantiation of previous results~~\cite{DBLP:journals/tcs/AbadiL91,DBLP:journals/iandc/LynchV95}.

\vspace{-1.5mm}
\begin{theorem}\label{th:forSim}
$L_1$ $\Gamma$-refines $L_2$ when there is a $\Gamma$-forward simulation from $L_1$ to $L_2$. Moreover, if $L_1$ $\Gamma$-refines $L_2$ and $L_2$ is $\Gamma$-deterministic, then there is a $\Gamma$-forward simulation from $L_1$ to $L_2$.
\vspace{-1.5mm}
\end{theorem}

The linearization of a concurrent history can be also defined in terms of \emph{linearization points}. Informally, a linearization point of 
an operation in an execution is a point in time where the operation is conceptually effectuated; given the linearization points of 
each operation, the linearization of a concurrent history is the sequential history which takes operations in order of their linearization points.
For some libraries, the linearization points correspond to a fixed set of actions. For instance, in the case of atomic libraries  
where method bodies are guarded with a global-lock acquisition, the linearization point of every method invocation corresponds to the execution 
of the body. When the linearization points are fixed, we assume that the library is an LTS over an alphabet that includes actions 
$lin(m,d,k)$ with $m\in\<Methods>$, $d\in\<Vals>$ and $k\in \<Ops>$. The action $lin(m,d,k)$ represents the linearization point of the operation $k$ 
returning value $d$.
Let $Lin$ denote the set of such actions. 
The projection of a library trace over $C\cup R\cup Lin$ is called an 
\emph{extended history}. A trace or extended history is called \emph{$Lin$-complete} when every completed operation has a linearization 
point, i.e., each return action $ret(m,d,k)$ is preceded by an action $lin(m,d,k)$. 
A library $L$ over alphabet $\Sigma$ is called \emph{with fixed linearization points} if{f} $C\cup R\cup Lin\subseteq \Sigma$ 
and every trace $@t\in Tr(L)$ is $Lin$-complete. 

Proving the correctness of an implementation $L_1$ of a concurrent object such as a queue or a stack with fixed linearization points
reduces to proving that $L_1$ is a $(C\cup R\cup Lin)$-refinement of an abstract implementation $L_2$ of the same object where method
bodies are guarded with a global-lock acquisition. Since the abstract implementation is usually $(C\cup R\cup Lin)$-deterministic,
by Theorem~\ref{th:forSim}, proving $(C\cup R\cup Lin)$-refinement is equivalent to finding a $(C\cup R\cup Lin)$-forward simulation 
from $L_1$ to $L_2$.

Section~\ref{sec:queues} and Section~\ref{sec:stacks} extend this result to queue and stack implementations where the linearization point of the methods 
\emph{adding} values to the collection is \emph{not} fixed.

\section{Queues With Fixed Dequeue Linearization Points}\label{sec:queues}
The typical abstract implementation of a concurrent queue, denoted as $AbsQ_0$, maintains a sequence of values, the enqueue adds a value atomically to the beginning of the sequence, and the dequeue removes a value from the end of the sequence (if any, otherwise it returns {\tt EMPTY}). Both methods have a fixed linearization point when the update of the sequence happens. %
For some queue implementations, e.g., the Herlihy\&Wing Queue~\cite{journals/toplas/HerlihyW90} ($\mathit{HWQ}$ for short), there exists no forward simulation to $AbsQ_0$ although they are a refinement of $AbsQ_0$. The main reason is that the enqueue methods don't have a \emph{fixed} linearization point. 
In this section, we propose a new abstract implementation for queues, denoted as $AbsQ$, which roughly maintains a \emph{partially-ordered set} of values instead of a sequence. We show that there exists a forward simulation from any correct queue implementation where only the \emph{dequeue} methods have fixed linearization points (the enqueue methods are unconstrained) to $AbsQ$. This covers all the queue implementations that we are aware of, in particular $\mathit{HWQ}$, Baskets Queue~\cite{DBLP:conf/opodis/HoffmanSS07}, LCRQ~\cite{DBLP:conf/ppopp/MorrisonA13}, or Time-Stamped Queue~\cite{DBLP:conf/popl/DoddsHK15} (where the enqueues don't have fixed linearization points). 
We also describe a forward simulation from $\mathit{HWQ}$ to $AbsQ$.

\vspace{-3.5mm}
\subsection{Enqueue Methods With Non-Fixed Linearization Points}
\vspace{-1mm}
We describe $\mathit{HWQ}$ where the linearization points of the enqueue methods are not fixed.
The shared state consists of an array {\tt items} storing the values in the queue and a counter {\tt back} storing the index of the first unused position in {\tt items}. Initially, all the positions in the array are {\tt null} and {\tt back} is 0.
An enqueue method starts by reserving a position in {\tt items} ({\tt i} stores the index of this position and {\tt back} is incremented so the same position can't be used by other enqueues) and then, stores the input value {\tt x} at this position. The dequeue method traverses the array {\tt items} starting from the beginning and atomically swaps {\tt null} with the encountered value. If the value is not {\tt null}, then the dequeue returns that value. If it reaches the end of the array, then it restarts.

\begin{wrapfigure}{l}{5.3cm}
\vspace{-9mm}
\begin{lstlisting}
void enq(int x){
  i = back++;
  items[i] = x;
}
int deq() {
  while (1) {
    range = back - 1;
    for (int i = 0; i <= range; i++){
      x = swap(items[i],null);
      if ( x != null ) return x;
}}}
  \end{lstlisting}
\vspace{-5.5mm}
\caption{Herlihy \& Wing Queue. We assume that every statement is atomic.}
\label{fig:HerlihyWing}
\vspace{-3mm}
\end{wrapfigure}
The linearization points of the enqueues are not fixed, they depend on dequeues executing in the future. Consider the following trace with two concurrent enqueues (${\tt i}(k)$ represents the value of {\tt i} in operation $k$): $inv(enq,x,1)$, $inv(enq,y,2)$, ${\tt i}(1) = \mbox{{\tt bck++}}$, ${\tt i}(2) = \mbox{{\tt bck++}}$, ${\tt items[i(}2{\tt )]} = y$.
Assuming that the linearization point corresponds to the assignment of {\tt i}, the history of this trace should be linearized to $inv(enq,x,1)$, $ret(enq,1)$, $inv(enq,y,2)$, $ret(enq,2)$. However, a dequeue executing until completion after this trace will return $y$ (only position $1$ is filled in the array {\tt items}) which is not consistent with this linearization. On the other hand, assuming that enqueues should be linearized at the assignment of {\tt items[i]} and extending the trace with ${\tt items[i(}1{\tt )]} = x$ and a completed dequeue that in this case returns $x$, leads to the incorrect linearization: $inv(enq,y,2)$, $ret(enq,2)$, $inv(enq,x,1)$, $ret(enq,1)$, $inv(deq,3)$, $ret(deq,x,3)$.

The dequeue method has a fixed linearization point which corresponds to an execution of {\tt swap} returning a non-null value. This action alone contributes to the effect of that value being removed from the queue. Every concurrent history can be linearized to a sequential history where dequeues occur in the order of their linearization points in the concurrent history.
This claim is formally proved in Section~\ref{ssec:HerlihyWing}.

Since the linearization points of the enqueues are not fixed, there exists no forward simulation from $\mathit{HWQ}$ to $AbsQ_0$. 
In the following, we describe the abstract implementation $AbsQ$ for which such a forward simulation does exist.

\subsection{Abstract Queue Implementation}
Informally, $AbsQ$ records the happens-before order between enqueue operations for which the added value has not been removed by a dequeue operation. The linearization point of a dequeue operation with return value $d\neq{\tt EMPTY}$ is enabled only if the happens-before stored in the current state contains a minimal enqueue that adds the value $d$. The effect of the linearization point is that the minimal enqueue is removed from the current state and the return value is recorded in the library state. When 
the return value is {\tt EMPTY}, the linearization point of a dequeue is enabled only if the current state stores only pending enqueues (the dequeue overlaps with all the enqueue operations stored in the current state and it can be linearized before all of them).  
The return of a dequeue is enabled only if the returned value matches the one fixed at the linearization point. 

\begin{wrapfigure}{l}{6.7cm}
\includegraphics[width=6.8cm]{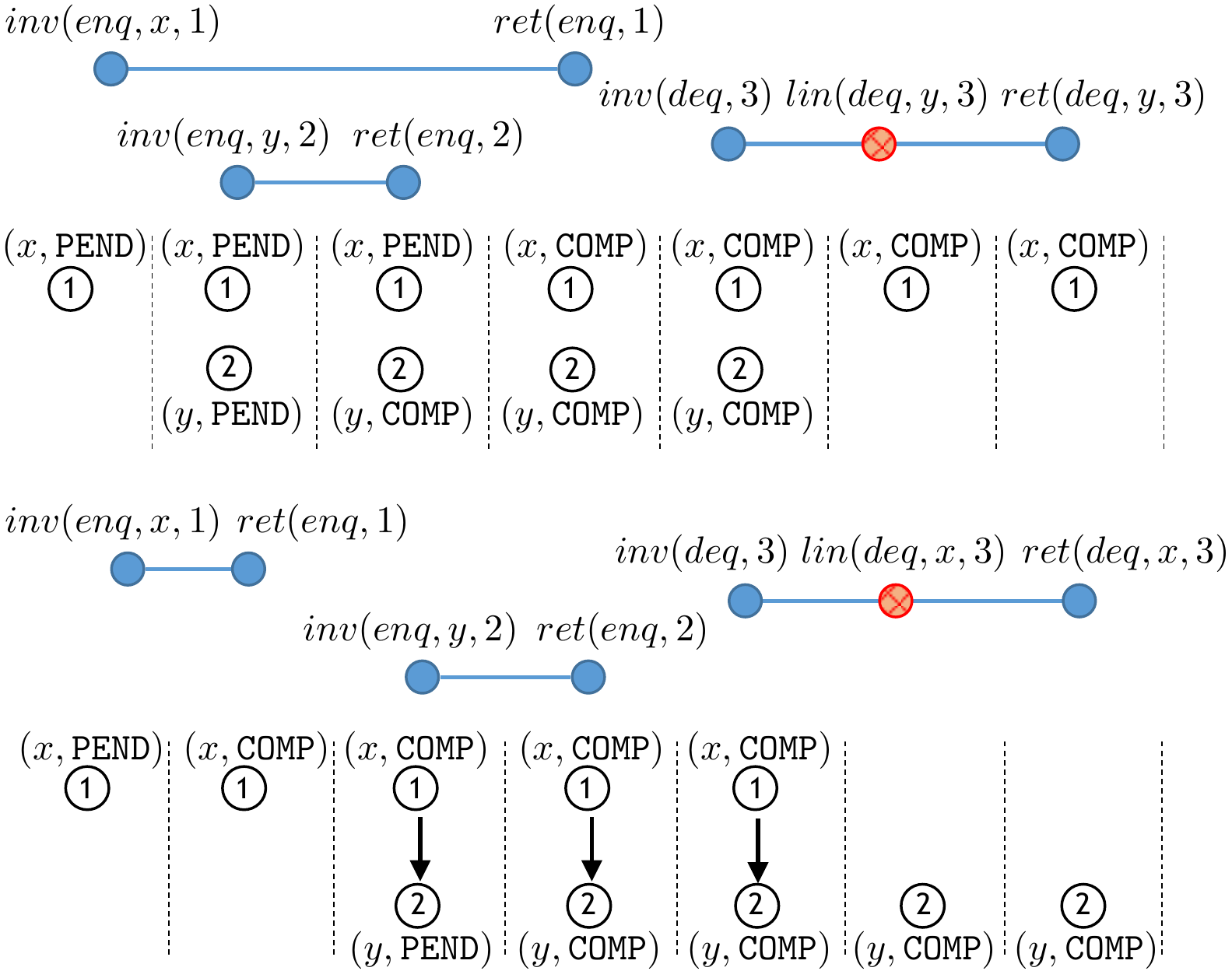}
\vspace{-8mm}
\caption{Simulating queue histories with $AbsQ$. The order between actions is from left to right.}
\label{fig:queueSim}
\vspace{-6mm}
\end{wrapfigure}
Figure~\ref{fig:queueSim} pictures two executions of $AbsQ$ for two extended histories (that include dequeue linearization points). The state of $AbsQ$ after each action is pictured as a graph below the action. The nodes of this graph represent enqueue operations and the edges happens-before constraints. Each node is labeled by a value (the input of the enqueue) and a flag {\tt PEND} or {\tt COMP} showing whether the operation is pending or completed. For instance, in the case of the first history, the dequeue linearization point $lin(deq,y,3)$ is enabled because the current happens-before contains a \emph{minimal} enqueue operation with input $y$. Note that a linearization point $lin(deq,x,3)$ is also enabled at this state.

Formally, the states of $AbsQ$ are tuples $\tup{O,<,\ell,rv,cp}$ where $O\subseteq \<Ops>$ is a set of operation identifiers, $<\subseteq O\times O$ is a strict partial order, $\ell: O -> \<Vals>\times\{\tt{PEND,\tt{COMP}}\}$ labels every identifier with a value and a pending/completed flag (the flag is used to track the happens-before order), $rv:\<Ops> ~> \<Vals>$ records the return value of a dequeue fixed at its linearization point ($~>$ denotes a partial function), and $cp:\<Ops> ~> \{A_1,A_2,R_1,R_2,R_3\}$ records the control point of every enqueue ($A_1, A_2$) or dequeue operation ($R_1,R_2,R_3$).
All the components are $\emptyset$ in the initial state, and the transition relation $->$ is defined in Fig.~\ref{fig:transitions:AbsQ}. The alphabet of $AbsQ$ contains call/return actions and dequeue linearization points, denoted by $lin(deq,d,k)$. $Lin(deq)$ is the set of all actions $lin(deq,d,k)$.

Concerning enqueue operations, the rule {\sc call-enq} orders the invoked operation after all the completed enqueues in the current state, and the rules {\sc ret-enq1}/{\sc ret-enq2} flip the corresponding flag from {\tt PEND} to {\tt COMP} provided that the operation is still present in the current state. For dequeue operations, {\sc call-deq} only increments the control point and {\sc ret-deq} checks whether the return value is the same as the one fixed at the linearization point. The linearization point rule {\sc lin-deq1} corresponds to the case of a non-empty queue, showing that $lin(deq,d,k)$ is enabled only if $d$ has been added by an enqueue which is minimal in the current happens-before. When enabled, it removes the enqueue adding $d$ from the state. The linearization point rule {\sc lin-deq2} corresponds to the case of dequeue operations linearized with an {\tt EMPTY} return value.

\begin{figure} [t]
{\scriptsize
  \centering
  \begin{mathpar}
    \inferrule[call-enq]{
      k\not\in dom(cp) \\ 
      d\neq {\tt EMPTY}
    }{
      O,<,\ell,rv,cp
      \xrightarrow{inv(enq,d,k)}
      O\cup\{k\},<\cup\ {\tt COMP}(O)\times\{k\},\ell[k\mapsto (d,{\tt PEND})],rv,cp[k\mapsto A_1]
    }\hspace{5mm}

    \inferrule[call-deq]{
      k\not\in dom(cp) \\ 
    }{
      O,<,\ell,rv,cp
      \xrightarrow{inv(deq,k)}
      O,<,\ell,rv,cp[k\mapsto R_1]
    }\hspace{5mm}
    \inferrule[ret-deq]{
       cp(k) = R_2 \\
       rv(k)=d  
    }{
      O,<,\ell,rv,cp
      \xrightarrow{ret(deq,d,k)}
      O,<,\ell,rv,cp[k\mapsto R_3]
    }\hspace{5mm}

    \inferrule[ret-enq1]{
      cp(k) = A_1 \\
      k \in O \\
      \ell(k) = (d,{\tt PEND}) 
    }{
      O,<,\ell,rv,cp
      \xrightarrow{ret(enq,k)}
      O,<,\ell[k\mapsto (d,{\tt COMP})],rv,cp[k\mapsto A_2]
    }\hspace{5mm}
    \inferrule[ret-enq2]{
      cp(k) = A_1 \\
      k \not\in O 
    }{
      O,<,\ell,rv,cp
      \xrightarrow{ret(enq,k)}
      O,<,\ell,rv,cp[k\mapsto A_2]
    }\hspace{5mm}

    \inferrule[lin-deq1]{
       cp(k) = R_1 \\
       d\neq{\tt EMPTY} \\
       k'\in min(O) \\ 
       \ell_1(k')=d
    }{
      O,<,\ell,rv,cp
      \xrightarrow{lin(deq,d,k)}
      O\setminus \{k'\},<\uparrow k',\ell,rv[k\mapsto d],cp[k\mapsto R_2]
    }\hspace{5mm}
    \inferrule[lin-deq2]{
       cp(k) = R_1 \\
       \forall o\in O.\ \ell_2(o)={\tt PEND}
    }{
      O,<,\ell,rv,cp
      \xrightarrow{lin(deq,{\tt EMPTY},k)}
      O,<,\ell,rv[k\mapsto {\tt EMPTY}],cp[k\mapsto R_2]
    }\hspace{5mm}    
      \end{mathpar}
  }
 \vspace{-6mm}
  \caption{The transition relation of $AbsQ$. We use the following notations: $\ell_i(k)$ denotes the projection of $\ell(k)$ over the $i$-th component, for each $i\in\{1,2\}$, ${\tt COMP}(O)=\{k\in O: \ell_2(k)={\tt COMP}\}$, $\mathit{f}[x\mapsto y]$ is the function $g$ such that $g(z)=f(z)$ for all $z\neq x$ in the domain of $f$, and $g(x)=y$, $min(O)$ is the set of elements of $O$ which are minimal in the order relation $<$, and $<\uparrow k$ denotes the relation $<$ where all the pairs containing $k$ have been removed.
  }
  \label{fig:transitions:AbsQ}
\vspace{-6mm}
\end{figure}

The following result states that the library $AbsQ$ has exactly the same set of histories as the standard abstract library $AbsQ_0$ (see Appendix~\ref{app:absImplQueue} for a proof).

\begin{theorem}\label{th:absImplQueue}
$AbsQ$ is a refinement of $AbsQ_0$ and vice-versa.
\end{theorem}

A trace of a queue implementation is called \emph{$Lin(deq)$-complete} when every completed dequeue has a linearization point, i.e., each return action $ret(deq,d,k)$ is preceded by an action $lin(deq,d,k)$. A queue implementation $L$ over alphabet $\Sigma$, such that $C\cup R\cup Lin(deq)\subseteq \Sigma$, is called \emph{with fixed dequeue linearization points} when every trace $@t\in Tr(L)$ is $Lin(deq)$-complete.

The following result shows that $C\cup R\cup Lin(deq)$-forward simulations are a sound and complete proof method for showing the correctness of a queue implementation with fixed dequeue linearization points (up to the correctness of the linearization points). It is obtained from Theorem~\ref{th:absImplQueue} and Theorem~\ref{th:forSim} using the fact that the alphabet of $AbsQ$ is exactly $C\cup R\cup Lin(deq)$ and $AbsQ$ is deterministic.

\begin{corollary}
A queue implementation $L$ with fixed dequeue linearization points is a $C\cup R\cup Lin(deq)$-refinement of $AbsQ_0$ if{f} there exists a $C\cup R\cup Lin(deq)$-forward simulation from $L$ to $AbsQ$.
\end{corollary}

\subsection{A Correctness Proof For Herlihy\&Wing Queue}\label{ssec:HerlihyWing}
We describe a forward simulation $\mathit{fs}_1$ from $\mathit{HWQ}$ to $AbsQ$. A $\mathit{HWQ}$ state is related by $\mathit{fs}_1$ to an $AbsQ$ state that consists of all the  enqueue operations for which the input is still present in the array {\tt items} and all the pending enqueue operations that have at most reserved an array position, ordered by a relation $<$ satisfying the following: 
\begin{itemize}
	\item[(a)] pending enqueues are maximal, i.e., for every two enqueues $k$ and $k'$ such that $k'$ is pending, we have that $k'\not< k$,
	\item[(b)] $<$ is consistent with the order in which positions of {\tt items} have been reserved, i.e., for every two enqueues $k$ and $k'$ such that ${\tt i}(k) < {\tt i}(k')$, we have that $k' \not< k$,
	\item[(c)] an enqueue which has reserved a position $i$ %
	can't be ordered before another enqueue that has reserved a position $j \geq i$ when the position $i$ has been ``observed'' by a non-linearized dequeue that may ``observe'' $j$ in the current array traversal, i.e., for every two enqueues $k$ and $k'$, and a dequeue $k_d$, such that 
	
	\vspace{-2mm}
	\noindent
	{\small
	\begin{align}
	\hspace{-8mm}
	{\tt x}(k_d)={\tt null} \land {\tt i}(k') \leq {\tt range}(k_d) \land {\tt i}(k) \leq {\tt i}(k_d) \leq {\tt i}(k')
	 \land ({\tt i}(k) = {\tt i}(k_d) => k_d\atCP {\tt if}\text{-}{\tt inc}) \label{eq:inst}
	\end{align}}
	
	\vspace{-6mm}
	\noindent
	we have that $k \not< k'$. The predicate $k_d\atCP {\tt if}\text{-}{\tt inc}$ holds when the dequeue $k_d$ is at a control point after a {\tt swap} returning {\tt null} and before the increment of {\tt i}.
\end{itemize}

\noindent
An enqueue is labeled by $(d,{\tt PEND})$ where $d$ is the input value if it's pending and by  $(d,{\tt COMP})$, otherwise. Also, for every dequeue operation $k$ such that ${\tt x}(k)=d\neq {\tt null}$, we have that $rv(k)=d$.

We show that $\mathit{fs}_1$ is indeed a $C\cup R\cup Lin(deq)$-forward simulation. Let $s$ and $t$ be states of $\mathit{HWQ}$ and $AbsQ$, respectively, such that $(s,t)\in\mathit{fs}_1$. 
We omit discussing the trivial case of transitions labeled by call and return actions which are simulated by similar transitions of $AbsQ$ (for the return a dequeue operation $k$, we use the equality between the local variable ${\tt x}(k)$ in $s$ and the component $rv(k)$ in $t$). 

We show that each internal step of an enqueue or dequeue, except the execution of {\tt swap} returning a non-null value in dequeue (which represents its linearization point), is simulated by an \emph{empty} sequence of $AbsQ$ transitions, i.e., for every state $s'$ obtained through one of these steps, if $(s,t)\in\mathit{fs}_1$, then $(s',t)\in\mathit{fs}_1$ for each $AbsQ$ state $t$. 
Essentially, this consists in proving the following property, called \emph{monotonicity}: the set of possible orders $<$ associated by $\mathit{fs}_1$ to $s'$ doesn't exclude any order $<$ associated to $s$.

Concerning enqueues, let $s'$ be the state obtained from $s$ when a pending enqueue $k$ reserves an array position. This enqueue must be maximal in both $t$ and any state $t'$ related to $s'$ (since it's pending). Moreover, there is no dequeue that can ``observe'' this position before restarting the array traversal. Therefore, item (c) in the definition of $<$ doesn't constrain the order between $k$ and some other enqueue neither in $s$ nor in $s'$. Since this transition doesn't affect the constraints on the order between enqueues different from $k$ (their local variables remain unchanged), monotonicity holds. This property is trivially satisfied by the second step of enqueue which doesn't affect {\tt i}.

To prove monotonicity in the case of dequeue internal steps different from its linearization point, it is important to track the non-trivial instantiations of item (c) in the definition of $<$ over the two states $s$ and $s'$, i.e., the triples $(k,k',k_d)$ for which (\ref{eq:inst}) holds. Instantiations that are enabled only in $s'$ may in principle lead to a violation of monotonicity (since they restrict the orders $<$ associated to $s'$). For the two steps that begin an array traversal, i.e., reading the index of the last used position and setting {\tt i} to $0$, there exist no  such new instantiations in $s'$ because the value of {\tt i} is either not set or $0$. %
The same is true for the increment of {\tt i} in a dequeue $k_d$ since the predicate $k_d\atCP {\tt if}\text{-}{\tt inc}$ holds in state $s$.
The execution of {\tt swap} returning {\tt null} in a dequeue $k_d$ enables new instantiations $(k,k',k_d)$ in $s'$, thus adding potentially new constraints $k\not< k'$. We show that these instantiations are however vacuous because $k$ must be pending in $s$ and thus maximal in every order $<$ associated by $\mathit{fs}_1$ to $s$.
Let $k$ and $k'$ be two enqueues such that together with the dequeue $k_d$ they satisfy the property (\ref{eq:inst}) in $s'$ but not in $s$. 
We write ${\tt i}_s(k)$ for the value of the variable {\tt i} of operation $k$ in state $s$. 
We have that ${\tt i}_{s'}(k) = {\tt i}_{s'}(k_d) \leq {\tt i}_{s'}(k')$ and ${\tt items}[{\tt i}_{s'}(k_d)]={\tt null}$. The latter implies that the enqueue $k$ didn't executed
the second statement (since the position it reserved is still {\tt null}) and it is pending in $s$. The step that checks that the value returned by {\tt swap} is {\tt null} doesn't modify the variables in property  (\ref{eq:inst}) and also, it doesn't change the valuation of the predicate $\atCP {\tt if}\text{-}{\tt inc}$.

Finally, we show that the linearization point of a dequeue $k$ of $\mathit{HWQ}$, i.e., an execution of {\tt swap} returning a non-null value $d$, from state $s$ and leading to a state $s'$ is simulated by a transition labeled by $lin(deq,d,k)$ of $AbsQ$ from state $t$. By the definition of $\mathit{HWQ}$, there is a unique enqueue $k_e$ which filled the position updated by $k$, i.e., ${\tt i}_s(k_e)=i_s(k)$ and ${\tt x}_{s'}(k)={\tt x}_s(k_e)$. We show that $k_e$ is minimal in the order $<$ of $t$ which implies that $lin(deq,d,k)$ is enabled in $t$. Thus, instantiating item (c) in the definition of $<$ with $k'=k_e$ and $k_d=k$ we get that every enqueue that reserved a position smaller than the one of $k_e$ can't be ordered before $k_e$ in the order $<$. Also, applying item (b) with $k=k_e$ we get the same for every enqueue that reserved a bigger position. An enqueue that didn't reserved a position is by definition maximal in $<$ and therefore, not a predecessor of $k_e$. Then, the state $t'$ obtained from $t$ through a $lin(deq,d,k)$ transition is related to $s'$ because (1) the value added by $k_e$ is not anymore present in {\tt items} which implies that $k_e$ doesn't occur in any $AbsQ$ state related to $s'$, and (2) the value of ${\tt x}(k)$ is set to $d\neq {\tt null}$ which implies that $rv(k)$ is set to $d$ in every $AbsQ$ state related to $s'$.
%
%
%
%
%
\section{Stacks With Fixed Pop Commit Points}\label{sec:stacks}
While the abstract queue in Section~\ref{sec:queues} can be adapted to stacks (the linearization point $lin(pop,d,k)$ with $d\neq{\tt EMPTY}$ is enabled when $k$ is added by a push which is maximal in the happens-before order stored in the state), it can't simulate (through forward simulations) existing stack implementations like the Time-Stamped Stack~\cite{DBLP:conf/popl/DoddsHK15} ($\mathit{TSS}$, for short) where the linearization points of the pop operations are not fixed. Exploiting particular properties of the stack semantics, we refine the ideas used in $AbsQ$ and define 
a new abstract implementation for stacks, denoted as $AbsQ$, which is able to simulate such implementations. Forward simulations to $AbsS$ are complete for proving the correctness of stack implementations provided that the point in time where the return value of a pop operation is determined, called \emph{commit point}, corresponds to a fixed action.
\subsection{Pop Methods With Fixed Commit Points}

We explain the meaning of the commit points on a simplified version of the Time-Stamped Stack~\cite{DBLP:conf/popl/DoddsHK15} ($\mathit{TSS}$, for short) given in Figure~\ref{fig:TimeStamped}. This  implementation maintains an array of singly-linked lists, one for each thread, where list nodes contain a data value (field {\tt data}), a timestamp (field {\tt ts}), the next pointer (field {\tt next}), and a boolean flag indicating whether the node represents a value removed from the stack (field {\tt taken}). Initially, each list contains a sentinel dummy node pointing to itself with timestamp $-1$ and the flag {\tt taken} set to {\tt false}.

\begin{wrapfigure}{l}{5.2cm}
\begin{lstlisting}
struct Node{
  int data;
  int ts;
  Node* next;
  boolean taken;
};
Node* pools[maxThreads];
int TS = 0;   

void push(int x) {
  Node* n = new Node(x,MAX_INT,
                        null,false);
  n->next = pools[myTID];
  pools[myTID] = n;
  int i = TS++;
  n->ts = i;
}
int pop() {
 boolean success = false;
 int maxTS = -1;
 Node* youngest = null;
 while ( !success ) {
   maxTS = -1; youngest = null;
   for(int i=0; i<maxThreads; i++){
     Node* n = pools[i];
     while (n->taken && n->next != n)
       n = n->next;
     if(maxTS < n->ts) {
       maxTS = n->ts; youngest = n;
     }
   }
   if (youngest != null)
     success=CAS(youngest->taken,
                       false,true);
 }
 return youngest->data;
}
\end{lstlisting}
\vspace{-6mm}
\caption{Time-Stamped Stack.} %
\label{fig:TimeStamped}
\vspace{-7mm}
\end{wrapfigure}
Pushing a value to the stack proceeds in several steps: adding a node with maximal timestamp in the list associated to the thread executing the push (given by the special variable {\tt myTID}), asking for a new timestamp (given by the shared variable {\tt TS}), and updating the timestamp of the added node. Popping a value from the stack consists in traversing all the lists, finding the first element which doesn't represent a removed value (i.e., {\tt taken} is {\tt false}) in each list, and selecting the element with the maximal timestamp. A compare-and-swap (CAS) is used to set the {\tt taken} flag of this element to {\tt true}. The procedure restarts if the CAS fails.

\begin{figure}[t]
\centering
\includegraphics[width=11.5cm]{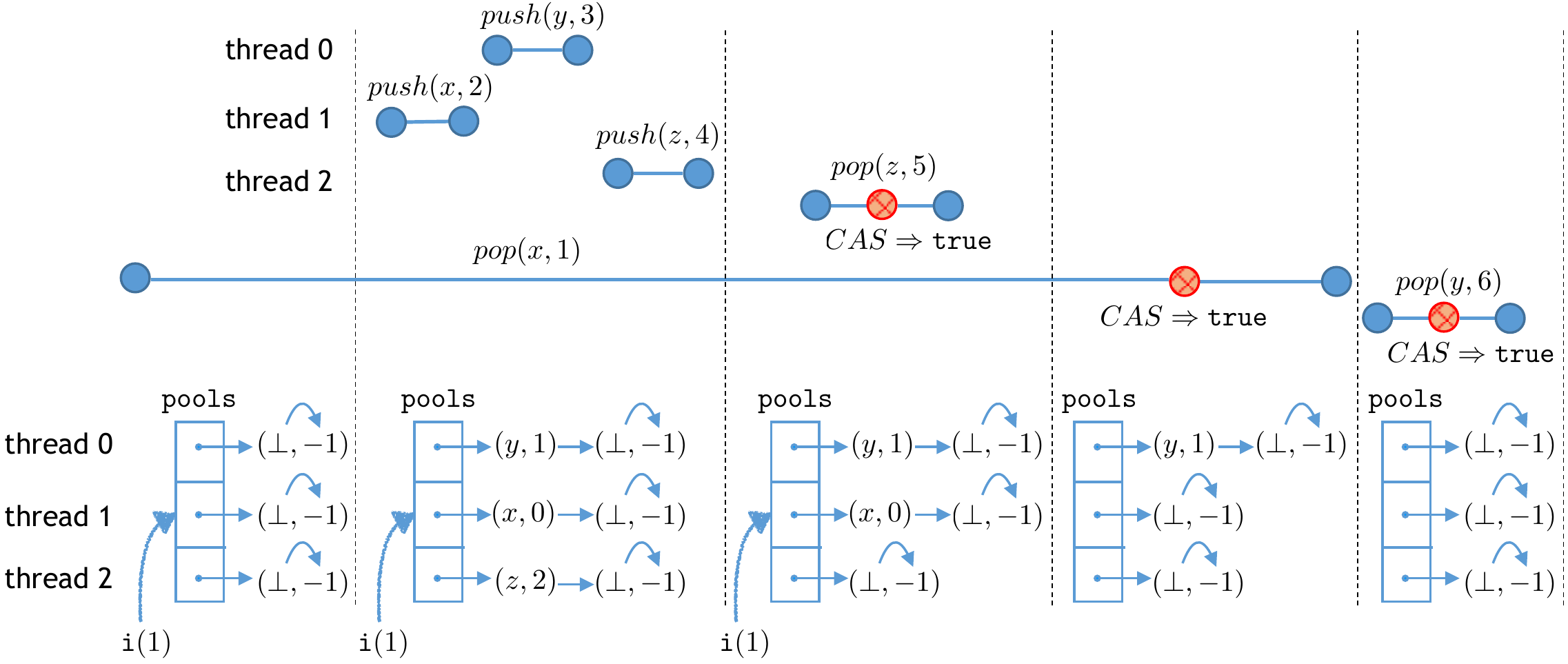}
\vspace{-4mm}
\caption{An execution of $\mathit{TSS}$. An operation is pictured by a line delimited by two circles denoting the call and respectively, the return action. Pop operations with identifier $k$ and removing value $d$ are labeled $pop(d,k)$. Their representation includes another circle that stands for a successful CAS which is their commit point. The library state after an execution prefix delimited at the right by a dotted line is pictured in the bottom part (the picture immediately to the left of the dotted line). A pair $(d,t)$ represents a list node with ${\tt data}=d$ and ${\tt ts}=t$, and ${\tt i}(1)$ denotes the value of {\tt i} in the pop with identifier 1. We omit the nodes where the field {\tt taken} is {\tt true}.}
\label{fig:commit}
\vspace{-6mm}
\end{figure}

The push operations don't have a fixed linearization point because adding a node to a list and updating its timestamp are not executed in a single atomic step. The nodes can be added in an order which is not consistent with the order between the timestamps assigned later in the execution. Also, the value added by a push that just added an element to a list can be popped before the value added by a completed push (since it has a maximal timestamp). The same holds for pop operations: The only reasonable choice for a linearization point is a successful CAS (that results in updating the field {\tt taken}). Fig.~\ref{fig:commit} pictures an execution showing that this action doesn't correspond to a linearization point, i.e., an execution for which the pop operations in every correct linearization are not ordered according to the order between successful CASs. In every correct linearization of that execution, the pop operation removing $x$ is ordered before the one removing $z$ although they perform a successful CAS in the opposite order.

An interesting property of the successful CASs in pop operations is that they fix the return value, i.e., the return value is {\tt youngest->data} where {\tt youngest} is the node updated by the CAS. We call such actions \emph{commit points}. More generally, commit points are actions that access shared variables, from which every control-flow path leads to the return control point and contains no more accesses to the shared memory (i.e., after a commit point, the return value is computed using only local variables).

When the commit points of pop operations are fixed to particular implementation actions (e.g., a successful CAS) we assume that the library is an LTS over an alphabet that contains actions $com(pop,d,k)$ with $d\in\<Vals>$ and $k\in\<Ops>$ (denoting the commit point of the pop with identifier $k$ and returning $d$). Let $Com(pop)$ be the set of such actions.

\vspace{-3mm}
\subsection{Abstract stack implementation}
\vspace{-1mm}
We define an abstract stack $AbsS$ over alphabet $C\cup R\cup Com(pop)$ that essentially, similarly to $AbsQ$, maintains the happens-before order of the pushes whose value has not been yet removed. Pops are treated differently since the commit points are not necessarily linearization points, intuitively, a pop can be linearized before its commit. Each pop operation starts by taking a snapshot of the greatest completed push operations in the happens-before order, and continuously tracks the push operations which are overlapping with it. The commit point $com(pop,d,k)$ with $d\neq {\tt EMPTY}$ is enabled only if $d$ was added by one of the push operations in the initial snapshot, or by a push happening earlier when all the values from the initial snapshot have been removed, or by one of the push operations that overlaps with pop $k$. The commit point $com(pop,{\tt EMPTY},k)$ is enabled only if all the values added by push operations ending before $k$ started have been removed. The effect of the commit points is explained below through examples.
\vspace{-.4mm}

\begin{wrapfigure}{l}{6.8cm}
\vspace{-6mm}
\includegraphics[width=7cm]{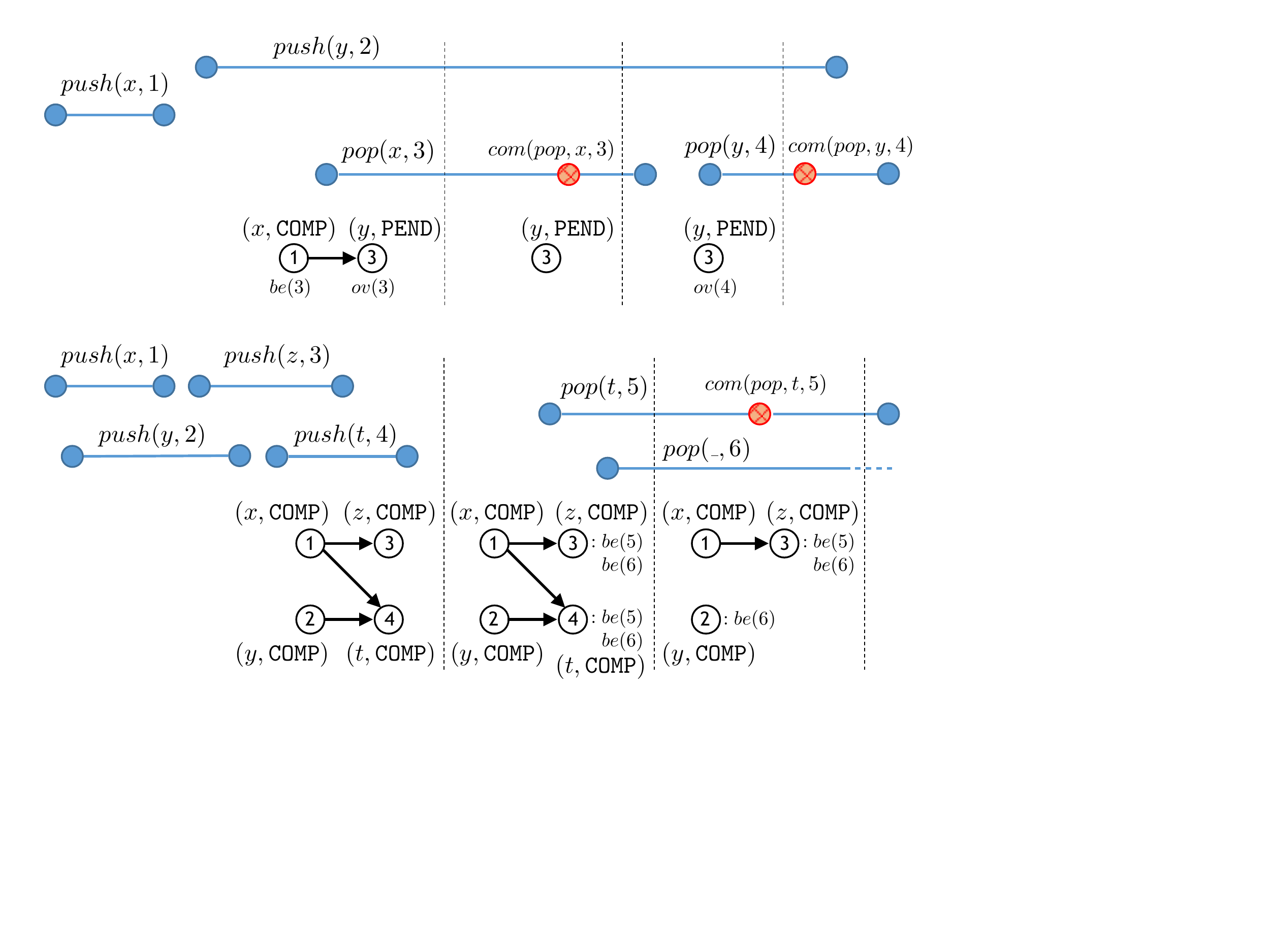}
\vspace{-8mm}
\caption{Simulating stack histories with $AbsS$.}
\label{fig:stackSim}
\vspace{-6.5mm}
\end{wrapfigure}
Figure~\ref{fig:stackSim} pictures two executions of $AbsS$ for two extended histories (that include pop commit points). For readability, we give the state of $AbsS$ only after several execution prefixes delimited at the right by a dotted line. We focus on pop operations -- the effect of push calls and returns is similar to enqueue calls and returns in $AbsQ$. Let us first consider the history on the top part. The first state we give is reached after the call of pop with identifier $3$. This shows the effect of a pop invocation: the greatest completed pushes according to the current happens-before (here, the push with identifier $1$) are marked as $be(3)$ (from ``before'' operation 3), and the pending pushes are marked as $ov(3)$ (from ``overlapping'' with operation 3). As a side remark, any other push operation that starts after pop $3$ would be also marked as $ov(3)$.
The commit point $com(pop,x,3)$ (pictured with a red circle) is enabled because $x$ was added by a push marked as $be(3)$. The effect of the commit point is that push $1$ is removed from the state (the execution on the bottom shows a more complicated case). For the second pop, the commit point $com(pop,y,4)$ is enabled because $y$ was added by a push marked as $ov(4)$. The execution on the bottom shows an example where the marking $be(k)$ for some pop $k$ is updated at commit points. The pushes $3$ and $4$ are marked as $be(5)$ and $be(6)$
when the pops $5$ and $6$ start. Then, $com(pop,t,5)$ is enabled since $t$ was added by $push(t,4)$ which is marked as $be(5)$. Besides removing $push(t,4)$, the commit point produces a state where a pop committing later, e.g., pop $6$, can remove $y$ which was added by a predecessor of $push(t,4)$ in the happens-before ($y$ could become the top of the stack when $t$ is removed). This history is valid because $push(y,2)$ can be linearized after $push(x,1)$ and $push(z,3)$. Thus, push 2, a predecessor of the push which is removed, is marked as $be(6)$. Push $1$ which is also a predecessor of the removed push is not marked as $be(6)$ because it happens before another push, i.e., push 3, which is already marked as $be(6)$ (the value added by push 3 should be removed before the value added by push 1 could become the top of the stack).

Formally, the states of $AbsS$ are tuples $\tup{O,<,\ell,rv,cp,be,ov}$ where $<$ is a strict partial order over the set $O$ of operation identifiers, $\ell: O -> \<Vals>\times\{\tt{PEND,\tt{COMP}}\}$ labels every identifier in $O$ with a value and a pending/completed flag, $rv:\<Ops> ~> \<Vals>$ records the return value of a pending pop fixed at its commit point, $cp:\<Ops> ~> \{A_1,A_2,R_1,R_2,R_3\}$ records the control point of every push ($A_1, A_2$) or pop operation ($R_1,R_2,R_3$), $be:\<Ops> ~> 2^O$ records the greatest completed push operations before a pop started or happening earlier provided that the values of all the push happening later have been removed, and $ov: \<Ops> ~> 2^O$ records push operations overlapping with a pop.
All the components are $\emptyset$ in the initial state, and the transition relation $->$ is defined in Fig.~\ref{fig:transitions:AbsS}.

The transition rules which don't correspond to commit point actions are similar to those for $AbsQ$. The rule {\sc com-pop1} for $com(pop,d,k)$ is enabled only if there exists a push $k'$ which added value $d$ and which belongs to $be(k)$ or $ov(k)$. When enabled, the push $k'$ is removed from the set $O$ (and the order $<$) and for every other pop $k_1$ such that $k'$ belongs to $be(k_1)$, $k'$ is replaced in $be(k_1)$ by its predecessors which are followed exclusively by pushes overlapping with $k_1$ (these predecessors become maximal closed pushes once $k'$ is removed). Also, $rv(k)$ is set to $d$. The rule {\sc com-pop1} for $com(pop,{\tt EMPTY},k)$ is enabled only if $be(k)$ is empty (i.e., all the values added by pushes ending before $k$, if any, have been removed). Then, $rv(k)$ is set to ${\tt EMPTY}$.

\begin{figure} [t]
\vspace{-2mm}
{\scriptsize
  \centering
  \begin{mathpar}
    \inferrule[call-push]{
      k\not\in dom(cp) \\ 
      d\neq {\tt EMPTY} \\
      \forall k'.\ ov'(k') = ov(k')\cup \{k\}
    }{
      O,<,\ell,rv,cp,be,ov 
      \xrightarrow{inv(push,d,k)} 
      O\cup\{k\},<\cup\ {\tt COMP}(O)\times\{k\},\ell[k\mapsto (d,{\tt PEND})],rv,cp[k\mapsto A_1],be,ov'
    }\hspace{5mm}

    \vspace{-1mm}
    \inferrule[call-pop]{
      k\not\in dom(cp) %
    }{
      O,<,\ell,rv,cp,be,ov
      \xrightarrow{inv(pop,k)} 
      O,<,\ell,rv,cp[k\mapsto R_1],be[k\mapsto maxCo(O)],ov[k\mapsto {\tt PEND}(O)]
    }\hspace{5mm}
    
    \vspace{-1mm}
    \inferrule[ret-pop]{
       cp(k) = R_2 \\
       rv(k)=d  
    }{
      O,<,\ell,rv,cp,be,ov
      \xrightarrow{ret(pop,d,k)}
      O,<,\ell,rv,cp[k\mapsto R_3],be,ov
    }\hspace{5mm}

    \vspace{-1mm}
    \inferrule[ret-push1]{
      cp(k) = A_1 \\
      k \in O \\
      \ell(k) = (d,{\tt PEND}) 
    }{
      O,<,\ell,rv,cp,be,ov
      \xrightarrow{ret(push,k)}
      O,<,\ell[k\mapsto (d,{\tt COMP})],rv,cp[k\mapsto A_2],be,ov
    }\hspace{5mm}
    \inferrule[ret-push2]{
      cp(k) = A_1 \\
      k \not\in O 
    }{
      O,<,\ell,rv,cp,be,ov
      \xrightarrow{ret(push,k)}
      O,<,\ell,rv,cp[k\mapsto A_2],be,ov
    }\hspace{5mm}

    \inferrule[com-pop1]{
       cp(k) = R_1 \\
       d\neq{\tt EMPTY} \\
       k'\in be(k)\cup ov(k) \\
        \ell_1(k')=d \\
       \forall k_1.\ k'\not\in be(k_1) \Rightarrow be'(k_1)=be(k_1) \\     \\
       \forall k_1.\ k'\in be(k_1) \Rightarrow be'(k_1)=(be(k_1)\setminus\{k'\})\cup \{k_2: k_2\in pred_{<}(k') \land \forall k_3. (k_2\in pred_{<}(k_3) \land k_3\neq k') => k_3\in ov(k_1)\} \\
    }{
      O,<,\ell,rv,cp,be,ov
      \xrightarrow{com(pop,d,k)} \\
      O\setminus \{k'\},<\uparrow k',\ell,rv[k\mapsto d],cp[k\mapsto R_2],be',ov %
    }\hspace{5mm}

    \inferrule[com-pop2]{
       cp(k) = R_1 \\
       be(k)=\emptyset
    }{
      O,<,\ell,rv,cp,be,ov
      \xrightarrow{com(pop,{\tt EMPTY},k)}
      O,<,\ell,rv[k\mapsto {\tt EMPTY}],cp[k\mapsto R_2],be,ov
    }\hspace{5mm}

      \end{mathpar}
  }
 \vspace{-6mm}
  \caption{The transition relation of $AbsQ$. We use the following notions: $maxCo(O)$ is the set of greatest operations in $O$ (w.r.t. $<$) which are completed, i.e., $maxCo(O)=\set{k\in O: \ell_2(k)={\tt COMP}, \forall k'\in O.\ k' < k \vee \ell_2(k')={\tt PEND}}$, ${\tt PEND}(O)=\{k\in O: \ell_2(k)={\tt PEND}\}$, and $pred_{<}(k')$ is the set of immediate predecessors of $k'$ according to $<$, i.e., $pred_{<}(k')=\set{k\in O: k < k'\land \forall k''\in O.\ k'' > k' \vee k'' < k}$.%
  }
  \label{fig:transitions:AbsS}
\vspace{-5mm}
\end{figure}

Let $AbsS_0$ be the standard abstract implementation of a stack (where elements are stored in a sequence; push, resp., pop operations add, resp., remove, an element from the beginning of the sequence in one atomic step). For $\<Methods>=\{push,pop\}$, the alphabet of $AbsS_0$ is $C\cup R\cup Lin$.
The following result states that the library $AbsS$ has exactly the same set of histories as $AbsS_0$ (see Appendix~\ref{app:absImplStack} for a proof).

\vspace{-2mm}
\begin{theorem}\label{th:absImplStack}
$AbsS$ is a refinement of $AbsS_0$ and vice-versa.
\vspace{-2mm}
\end{theorem}

A trace of a stack implementation is called \emph{$Com(pop)$-complete} when every completed pop has a commit point, i.e., each return $ret(pop,d,k)$ is preceded by an action $com(pop,d,k)$. A stack implementation $L$ over $\Sigma$, such that $C\cup R\cup Com(pop)\subseteq \Sigma$, is called \emph{with fixed pop commit points} when every trace $@t\in Tr(L)$ is $Com(pop)$-complete.

As a consequence of Theorem~\ref{th:forSim}, $C\cup R\cup Com(pop)$-forward simulations are a sound and complete proof method for showing the correctness of a stack implementation with fixed pop commit points (up to the correctness of the commit points).

\vspace{-1.5mm}
\begin{corollary}
A stack $L$ with fixed pop commit points is a $C\cup R\cup Com(pop)$-refinement of $AbsS$ if{f} there is a $C\cup R\cup Com(pop)$-forward simulation from $L$ to $AbsS$.
\vspace{-1.5mm}
\end{corollary}

Linearization points can also be seen as commit points and thus the following holds.

\vspace{-1.5mm}
\begin{corollary}
A stack implementation $L$ with fixed pop linearization points where transition labels $lin(pop,d,k)$ are substituted with $com(pop,d,k)$ is a $C\cup R\cup Com(pop)$-refinement of $AbsS_0$ if{f} there is a $C\cup R\cup Com(pop)$-forward simulation from $L$ to $AbsS$.
\vspace{-1.5mm}
\end{corollary}

\vspace{-6mm}
\subsection{A Correctness Proof For Time-Stamped Stack}\label{sec:corr_tss}
\vspace{-1mm}
We describe a forward simulation $\mathit{fs}_2$ from $\mathit{TSS}$ to $AbsS$. Except for the constraints on the components $be$ and $ov$ of a $AbsS$ state, it is similar to the simulation $\mathit{fs}_1$ from $\mathit{HWQ}$ to $AbsQ$. Thus, the $AbsS$ states $t=\tup{O,<,\ell,rv,cp,be,ov}$ associated by $\mathit{fs}_2$ to a $\mathit{TSS}$ state $s$ satisfy the following. The set $O$ consists of all the identifiers of pushes in $s$ which didn't added yet a node to {\tt pools} or for which the input is still present in {\tt pools} (i.e., the node created by the push has {\tt taken} set to {\tt false}). A push $k$ is labeled by $(d,{\tt PEND})$ where $d$ is the input value if it's pending and by $(d,{\tt COMP})$, otherwise. 

To describe the order relation $<$ we consider the following notations: ${\tt ts}_s(k)$, resp., ${\tt TID}_s(k)$, denotes the timestamp of the node created by the push $k$ in state $s$ (the {\tt ts} field of this node), resp., the id of the thread executing $k$. By an abuse of terminology, we call ${\tt ts}_s(k)$ the timestamp of $k$ in state $s$.
Also, $k \leadsto_s k'$ when intuitively, a traversal of {\tt pools}  would encounter the node created by $k$ before the one created by $k'$. More precisely, $k \leadsto_s k'$ when ${\tt TID}_s(k) < {\tt TID}_s(k')$, or ${\tt TID}_s(k) = {\tt TID}_s(k')$ and the node created by $k'$ is reachable from the one created by $k$ in the list pointed to by ${\tt pools}[{\tt TID}_s(k)]$.
The order relation $<$ satisfies the following: (1) pending pushes are maximal, (2) $<$ is consistent with the order between node timestamps, i.e., ${\tt ts}_s(k) \leq {\tt ts}_s(k')$ implies $k'\not< k$, and (3) $<$ includes the order between pushes executed in the same thread, i.e., ${\tt TID}_s(k) = {\tt TID}_s(k')$ and ${\tt ts}_s(k) < {\tt ts}_s(k')$ implies $k < k'$.

The components $be$ and $ov$ satisfy the following constraints (their domain is the set of identifiers of pending pops):
\vspace{-2mm}
\begin{itemize}
	\item a pop $k$ with ${\tt youngest}\neq {\tt null}$ that reached a node with timestamp $\tau$ (its variable {\tt n} points to this node) overlaps with every push that created a node with a timestamp bigger than $\tau$ and which occurs in {\tt pools} before the node reached by $k$, i.e., ${\tt youngest}_s(k)\neq{\tt null}$, ${\tt n}_s(k)={\tt n}_s(k_1)$, $k_2\leadsto_s k_1$, ${\tt n}_s(k_2)\text{\tt ->taken}={\tt false}$, and ${\tt ts}_s(k_2) \geq {\tt ts}_s(k_1)$ implies $k_2\in ov(k)$, for each $k, k_1, k_2$
	\item a pop $k$ with ${\tt youngest}={\tt null}$ overlaps with every push that created a node which occurs in {\tt pools} before the node reached by $k$,
i.e., ${\tt youngest}_s(k)={\tt null}$, ${\tt n}_s(k)={\tt n}_s(k_1)$, $k_2\leadsto_s k_1$, and ${\tt n}_s(k_2)\text{\tt ->taken}={\tt false}$ implies $k_2\in ov(k)$, for each $k, k_1, k_2$
	\item if the variable {\tt youngest} of a pop $k$ points to a node which is not taken, then this node was created by a push in $be(k)\cup ov(k)$ or the node currently reached by $k$ is followed in {\tt pools} by another node which was created by a push in $be(k)\cup ov(k)$, i.e., ${\tt youngest}_s(k)={\tt n}_s(k_1)$, ${\tt n}_s(k_1)\text{{\tt ->taken}}={\tt false}$, and ${\tt n}_s(k)={\tt n}_s(k_2)$ implies $k_1\in be(k)\cup ov(k)$ or that there exists $k_3\in O$ such that ${\tt ts}_s(k_3) > {\tt ts}_s(k_1)$, $k_3\in be(k)\cup ov(k)$, and either $k_2\leadsto_s k_3$ or ${\tt n}_s(k_2)={\tt n}_s(k_3)$ and TODO $k$ is traversing the last list in the array {\tt pools}, for each $k, k_1,k_2$
\vspace{-2mm}
\end{itemize}
There are some more constraints on $be$ and $ov$ that can be seen as invariants of $AbsS$, i.e., $be(k)$ and $ov(k)$ don't contain predecessors of pushes from $be(k)$ (for each $k, k_1, k_2$, $k_1 < k_2$ and $k_2 \in be(k)$ implies $k_1\not\in be(k)\cup ov(k)$). They can be found in Appendix~\ref{app:tss}.

Finally, for every pop operation $k$ such that ${\tt success}(k)={\tt true}$, we have that $rv(k)={\tt youngest}(k)\text{\tt->data}$. 

The proof that $\mathit{fs}_2$ is indeed a forward simulation from $\mathit{TSS}$ to $AbsS$ follows the same lines as the one given for the Herlihy\&Wing Queue. It can be found in Appendix~\ref{app:tss}.

\vspace{-3.5mm}
\section{Related Work}
\vspace{-1.5mm}
Many techniques for linearizability verification, e.g.,~\cite{conf/ppopp/VafeiadisHHS06,conf/cav/AmitRRSY07,conf/vmcai/Vafeiadis09,conf/tacas/AbdullaHHJR13}, are based on forward simulation arguments, and typically only work for libraries where the linearization point of every invocation of a method $m$ is fixed to a particular statement in the code of $m$. The works in~\cite{conf/cav/Vafeiadis10,Derrick2011,conf/cav/DragoiGH13,DBLP:conf/cav/ZhuPJ15} deal with \emph{external} linearization points where the action of an operation $k$ can be the linearization point of a concurrently executing operation $k'$. We say that the linearization point of $k'$ is external. This situation arises in read-only methods like the {\tt contains} method of an optimistic set~\cite{conf/podc/OHearnRVYY10}, libraries based on the elimination back-off scheme, e.g.,~\cite{conf/spaa/HendlerSY04}, or flat combining~\cite{DBLP:conf/spaa/HendlerIST10,DBLP:conf/podc/GorelikH13}. 
In these implementations, an operation can do an update on the shared state that becomes the linearization point of a concurrent read-only method (e.g., a {\tt contains} returning {\tt true} may be linearized when an {\tt add} method adds a new value to the shared state) or an operation may update the data structure on behalf of other concurrently executing operations (whose updates are published in the shared state). In all these cases, every linearization point can still be associated syntactically to a statement in the code of a method and doesn't depend on operations executed in the future (unlike $\mathit{HWQ}$ and $\mathit{TSS}$). However, identifying the set of operations for which such a statement is a linearization point can only be done by looking at the whole program state (the local states of all the active operations). This poses a problem in the context of compositional reasoning (where auxiliary variables are required), but still admits a forward simulation argument. For manual proofs, such implementations with external linearization points can still be defined as LTSs that produce $Lin$-complete traces and thus still fall in the class of implementations for which forward simulations are enough for proving refinement. These proof methods are not complete and they are not able to deal with implementations like $\mathit{HWQ}$ or $\mathit{TSS}$.

There also exist linearizability proof techniques based on backward simulations or alternatively, prophecy variables, e.g.,~\cite{phd/Vafeiadis08,DBLP:conf/cav/SchellhornWD12,DBLP:conf/pldi/LiangF13}. These works can deal with implementations where the linearization points are not fixed, but the proofs are conceptually more complex and less amenable to automation.

The works in~\cite{conf/concur/HenzingerSV13,DBLP:conf/icalp/BouajjaniEEH15} propose reductions of linearizability to assertion checking where the idea is to define finite-state automata that recognize violations of concurrent queues and stacks. These automata are simple enough in the case of queues and there is a proof of $\mathit{HWQ}$ based on this reduction~\cite{conf/concur/HenzingerSV13}. However, in the case of stacks, the automata become much more complicated and we are not aware of a proof for an implementation such as $\mathit{TSS}$ which is based on this reduction.
 

\newpage
\appendix

\section{Libraries}\label{app:prelim}

Programs interact with libraries by calling named library \emph{methods}, which
receive \emph{parameter values} and yield \emph{return values} upon completion.
We fix arbitrary sets $\<Methods>$ and $\<Vals>$ of method names and
parameter/return values.

\noindent
We fix an arbitrary set $\<Ops>$ of operation identifiers, and for given sets
$\<Methods>$ and $\<Vals>$ of methods and values, we fix the sets
\begin{align*}
  & C = \set{ inv(m,d,k) : m \in \<Methods>, d \in \<Vals>, k \in \<Ops> }
  \text{, and } \\
  & R = \set{ ret(m,d,k) : m \in \<Methods>, d \in \<Vals>, k \in \<Ops> }  
\end{align*}
of \emph{call actions} and \emph{return actions}; each call action $inv(m,d,k)$
combines a method $m \in \<Methods>$ and value $d \in \<Vals>$ with an
\emph{operation identifier} $k \in \<Ops>$. Operation identifiers are used to
pair call and return actions. 
We assume every set of 
words is closed under isomorphic renaming of operation identifiers. 
We denote the operation identifier of a
call/return action $a$ by $\<op>(a)$. Call and return actions $c \in C$ and $r
\in R$ are \emph{matching}, written $c \match r$, when $\<op>(c) = \<op>(r)$. 
We may omit the second field from a call/return action $a$ for methods that have no inputs (e.g., the pop method of a stack) or return values (e.g., the push method of a stack).
A word $\tau \in @S^*$ over alphabet $@S$, such that $(C \u R) \subseteq @S$, is
\emph{well formed} when:
\begin{itemize}

  \item Each return is preceded by a matching call: \\
  $\tau_j \in R$ implies $\tau_i \match \tau_j$ for some $i < j$.

  \item Each operation identifier is used in at most one call/return: \\
  $\<op>(\tau_i) = \<op>(\tau_j)$ and $i < j$ implies $\tau_i \match \tau_j$.

\end{itemize}
We say that the well-formed word $\tau \in @S^*$ is \emph{sequential} when
\begin{itemize}

  \item Operations do not overlap: \\
  $\tau_i, \tau_k \in C$ and $i < k$ implies $\tau_i \match \tau_j$ for some $i < j < k$.

\end{itemize}
Well-formed words represent traces of a library. We assume every set of well-formed
words is closed under isomorphic renaming of operation identifiers. For
notational convenience, we take $\<Ops>=\<Nats>$ for the rest of the paper.
When the value of a certain field in a call/return action is not important we use 
the placeholder $\_$, e.g., $inv(m,\_,k)$ instead of $inv(m,d,k)$ when the input  
$d$ can take any value.

An operation $k$ is called \emph{completed} in a well-formed trace $\tau$ when
$ret(m,d,k)$ occurs in $\tau$, for some $m$ and $d$. Otherwise, it is called \emph{pending}.

Libraries dictate the execution of methods between their call and return
points. Accordingly, a library cannot prevent a method from being called,
though it can decide not to return. Furthermore, any library action performed
in the interval between call and return points can also be performed should the
call have been made earlier, and/or the return made later. 
A library thus allows any sequence of
invocations to its methods made by \emph{some} program.

\begin{definition}\label{def:libraries}
A \emph{library} $L$ is an LTS over alphabet $\Sigma$ such that $C \u R\subseteq \Sigma$
and each trace $\tau \in Tr(L)$ is well formed, and
  \begin{itemize}

    \item Call actions $c \in C$ cannot be disabled: \\
    $\tau \cdot \tau' \in Tr(L)$ implies $\tau \cdot c \cdot \tau' \in Tr(L)$
    if $\tau \cdot c \cdot \tau'$ is well formed.
  
    \item Call actions $c \in C$ cannot disable other actions: \\
    $\tau \cdot a \cdot c \cdot \tau' \in Tr(L)$ implies $\tau \cdot c \cdot a \cdot \tau' \in Tr(L)$.
  
    \item Return actions $r \in R$ cannot enable other actions: \\
    $\tau \cdot r \cdot a \cdot \tau' \in Tr(L)$ implies $\tau \cdot a \cdot r \cdot \tau' \in Tr(L)$.
  
  \end{itemize}

\end{definition}

Note that even a library that implements \emph{atomic methods}, e.g.,~by
guarding method bodies with a global-lock acquisition, admits executions in
which method calls and returns overlap. 
For simplicity, Definition~\ref{def:libraries} assumes that every thread performs a single operation. The extension to multiple operations per thread is straightforward, e.g. the closure rules must assume that the actions $a$ and $c$ belong to different threads
\section{Normal Forward/Backward Simulations}\label{app:backSim}

We define a class of forward/backward simulations, called \emph{normal simulations}, that are used in the proofs in Appendix~\ref{app:absImplQueue} and Appendix~\ref{app:absImplStack}. 

\begin{definition}\label{def:for_app}
Let $L_1=(Q_1,\Sigma, s_0^1, \delta_1)$ and $L_2=(Q_2,\Sigma, s_0^2, \delta_2)$ be two libraries over alphabets $\Sigma_1$ and $\Sigma_2$, respectively, such that $C\cup R \subseteq \Sigma_1\cap\Sigma_2$, and $\Gamma$ a set of actions such that $C\cup R\subseteq \Gamma\subseteq \Sigma_1\cap\Sigma_2$. A relation $\mathit{fs} \subseteq Q_{1} \times Q_{2}$ is called a \emph{normal $\Gamma$-forward simulation} from $L_1$ to $L_2$ iff the following holds:
\begin{itemize}
\item[(i)] $\mathit{fs}[s_0^1] = \{s_0^2 \}$ 
\item[(ii-a)] If $(s,c,s') \in \delta_1$, for some $c\in C$, and $u \in \mathit{fs}[s]$, then there exists $u' \in \mathit{fs}[s']$ such that $u \xrightarrow{@s} u'$, $@s_0=c$, and $@s_i\in \Sigma_2\setminus\Gamma$, for each $0<i<|@s|$.
\item[(ii-b)] If $(s,r,s') \in \delta_{1}$, for some $r\in R$, and $u \in \mathit{fs}[s]$, then there exists $u' \in \mathit{fs}[s']$ such that $u \xrightarrow{@s} u'$, $@s_{|@s| -1}=r$, and $@s_i\in \Sigma_2\setminus\Gamma$, for each $0\leq i<|@s| -1$.
\item[(ii-c)] If $(s, \gamma , s') \in \delta_1$, for some $\gamma\in \Gamma\setminus (C\cup R)$, and $u \in fs[s]$, then there exists $u' \in fs[s']$ such that $\delta_2(u,\gamma, u')$. 
\item[(ii-d)] If $(s,e,s') \in \delta_1$, for some $e \in \Sigma_1\setminus \Gamma$ and $u \in \mathit{fs}[s]$, then there exists $u' \in \mathit{fs}[s']$ such that $u \xrightarrow{\sigma} u'$ and $\sigma\in (\Sigma_2\setminus\Gamma)^*$.  
\end{itemize}
\end{definition}
With normal $\Gamma$-forward simulations, a step of $L_1$ labeled by a call, resp., return, action is simulated by a sequence of steps of $L_2$ that start, resp., end, with the same action, and a step of $L_1$ labeled by another observable action should be matched by a step of $L_2$ labeled by the same action. The rest of the transitions in $L_1$ are matched to a possibly empty sequence of transitions of $L_2$ with arbitrary labels.

A dual notion of forward simulation is the backward simulation:
\begin{definition}\label{def:back_app}
Let $L_1=(Q_1,\Sigma, s_0^1, \delta_1)$ and $L_2=(Q_2,\Sigma, s_0^2, \delta_2)$ be two libraries over a common alphabet $\Sigma$, and $\Gamma\subseteq \Sigma$ a set of actions such that $(C\cup R)\subseteq \Gamma$. A relation $bs \subseteq Q_1 \times Q_2$ is called a \emph{normal $\Gamma$-backward simulation} from $L_1$ to $L_2$ iff the following holds:
\begin{itemize}
\item[(i)] $bs[s_0^1] = \{s_0^2 \}$
\item[(ii-a)] If $(s,c,s') \in \delta_1$, for some $c\in C$, and $u' \in bs[s']$, then there exists $u \in bs[s]$ such that $u \xrightarrow{@s} u'$, $@s_0=c$, and $@s_i\in \Sigma\setminus\Gamma$, for each $0<i<|@s|$.
\item[(ii-b)] If $(s,r,s') \in \delta_1$, for some $r\in R$, and $u' \in bs[s']$, then there exists $u \in bs[s]$ such that $u \xrightarrow{@s} u'$, $@s_{|@s| -1}=r$, and $@s_i\in \Sigma\setminus\Gamma$, for each $0\leq i<|@s| -1$.
\item[(ii-c)] If $(s,\gamma, s') \in \delta_1$, for some $\gamma\in \Gamma\setminus (C\cup R)$, and $u' \in bs[s']$, then there exists $u \in bs[s]$ such that $\delta_2(u,\gamma,u')$
\item[(ii-d)] If $(s,e,s') \in \delta_1$ for some $e \in \Sigma\setminus \Gamma$ and $u' \in bs[s']$, then there exists $u \in bs[s]$ such that $u \xrightarrow{@s} u'$ and $\sigma\in (\Sigma_2\setminus\Gamma)^*$.
\end{itemize}
\end{definition}

\section{Proof of Theorem~\ref{th:absImplQueue}}\label{app:absImplQueue}

\begin{figure} [t]
{\scriptsize
  \centering
  \begin{mathpar}
    \inferrule[call-enq]{
      k\not\in dom(cp^0) \\ 
      d\neq {\tt EMPTY}
    }{
      \sigma,in^0,rv^0,cp^0
      \xrightarrow{inv(enq,d,k)}
      \sigma,in^0[k\mapsto d], rv^0,cp^0[k\mapsto A_1]
    }\hspace{5mm}
    \inferrule[lin-enq]{
      cp^0(k)=A_1
    }{
      \sigma,in^0,rv^0,cp^0
      \xrightarrow{lin(enq,d,k)}
      d\cdot\sigma,in^0,rv^0,cp^0[k\mapsto A]
    }\hspace{5mm}
    
        \inferrule[ret-enq]{
      cp^0(k)=A
    }{
      \sigma,in^0,rv^0,cp^0
      \xrightarrow{ret(enq,k)}
      \sigma,in^0,rv^0,cp^0[k\mapsto A_2]
    }\hspace{5mm}

    \inferrule[call-deq]{
      k\not\in dom(cp^0) \\ 
    }{
      \sigma,in^0,rv^0,cp^0
      \xrightarrow{inv(deq,k)}
      \sigma,in^0,rv^0,cp^0[k\mapsto R_1]
    }\hspace{5mm}
        \inferrule[lin-deq1]{
      cp^0(k)=R_1 \\
      \sigma = \sigma'\cdot d 
    }{
      \sigma,in^0,rv^0,cp^0
      \xrightarrow{lin(deq,d,k)}
      \sigma',in^0,rv^0[k\mapsto d],cp^0[k\mapsto R_2]
    }\hspace{5mm}

        \inferrule[lin-deq2]{
      cp^0(k)=R_1 \\
      \sigma = \epsilon
    }{
      \sigma,in^0,rv^0,cp^0
      \xrightarrow{lin(deq,{\tt EMPTY},k)}
      \sigma,in^0,rv^0[k\mapsto {\tt EMPTY}],cp^0[k\mapsto R_2]
    }\hspace{5mm}
    \inferrule[ret-deq]{
      cp^0(k)=R_2 \\
      rv^0(k) = d
    }{
      \sigma,in^0,rv^0,cp^0
      \xrightarrow{ret(deq,d,k)}
      \sigma,in^0,rv^0,cp^0[k\mapsto R_3]
    }\hspace{5mm}
          \end{mathpar}
  }
 \vspace{-4mm}
  \caption{The transition relation of $AbsQ_0$. 
  }
  \label{fig:transitions:AbsQ_0}
\vspace{-2mm}
\end{figure}

We show that $AbsQ$ and $AbsQ_0$ refine each other. We start by giving a formal definition of the standard reference implementation $AbsQ_0$.
Thus, the states of $AbsQ_0$ are tuples $\tup{\sigma,in^0,rv^0,cp^0}$ where $\sigma\in\<Vals>^*$ is a sequence of values, $in^0:\<Ops> ~> \<Vals>$ records the input value of an enqueue, $rv^0:\<Ops> ~> \<Vals>$ records the return value of a dequeue fixed at its linearization point ($~>$ denotes a partial function), and $cp^0:\<Ops> ~> \{A_1,A,A_2,R_1,R_2,R_3\}$ records the control point of every enqueue ($A_1, A,A_2$) or dequeue operation ($R_1,R_2,R_3$).
All the components are $\emptyset$ in the initial state, and the transition relation $->$ is defined in Fig.~\ref{fig:transitions:AbsQ_0}. The alphabet of $AbsQ$ contains call/return actions and enqueue/dequeue linearization points.

To prove that $AbsQ$ is a refinement of $AbsQ_0$ we define a normal $C\cup R\cup Lin(deq)$-backward simulation (i.e, a backward simulation as in Definition~\ref{def:back_app}) from $AbsQ$ to $AbsQ_0$. The reverse is shown using a normal $C\cup R\cup Lin(deq)$-forward simulation (i.e, a forward simulation as in Definition~\ref{def:for_app}).

\begin{lemma} 
$AbsQ$ is a refinement of $AbsQ_0$.
\end{lemma}
\begin{proof}
We define a normal $C\cup R\cup Lin(deq)$-backward simulation $bs$ from $AbsQ$ to $AbsQ_0$ as follows. Given an $AbsQ$ state $s=\tup{O,<,\ell,rv,cp}$ and an $AbsQ_0$ state $t=\tup{\sigma,in^0,rv^0,cp^0}$ we have that $(s,t)\in bs$ iff the following hold:
\begin{itemize}
	\item the sequence $\sigma$ is a linearization of a partial order $(D,\prec)$ where $D$ contains values labeling elements of $O$ and all the values corresponding to completed enqueues, i.e., $\ell_1({\tt COMP}(O))\subseteq D\subseteq \ell_1(O)$ ordered according to the happens-before order between the enqueues that added them, i.e., $d_1\prec d_2$ if{f} there exists $k_1,k_2$ such that $\ell_1(k_1)=d_1$, $\ell_1(k_2)=d_2$, and $k_1 < k_2$.
	\item the return values fixed at dequeue linearization points are the same, i.e., for every $k$, $rv(k)=rv^0(k)$,
	\item every dequeue is at the same control point in both $s$ and $t$, i.e., for every $k$ and $i\in \{1,2,3\}$, $cp(k)=R_i$ iff $cp^0(k)=R_i$,
	\item every pending enqueue has the same input value in both $s$ and $t$, i.e., for every $k$, $\ell_1(k)=in^0(k)$,
	\item a pending enqueue from $O$ has been linearized whenever its value is contained in $\sigma$, i.e., for every $k$, $cp^0(k)=A$ if $\ell_1(k)\in D$ and $\ell_2(k)={\tt PEND}$, 
	\item a pending enqueue from $O$ hasn't been linearized whenever its value is not in $\sigma$, i.e., for every $k$, $cp^0(k)=A_1$ iff $\ell_1(k)\not\in D$ and $\ell_2(k)={\tt PEND}$, 
	\item a pending enqueue which is not in $O$ has been linearized, i.e., for every $k$, $cp^0(k)=A$ if $k\not\in O$ and $cp(k)=A_1$, 
	\item an enqueue is completed in $s$ whenever it is completed in $t$, i.e., for every $k$, $cp(k)=A_2$ iff $cp^0(k)=A_2$,
\end{itemize}

For the conditions described above, if we fix the set $D$ and $\sigma_t$, then the state $t$ related to $s$ becomes unique. We use this fact in the proof. In some places, we only give $D$, $\sigma_t$ and $s$ without explicitly defining $t$ or show that there exists $t$ with the given $\sigma_t$ that is related to $s$ by just finding a $D$ such that $\sigma_t$ is a linearization of $(D, \prec)$ where $\prec$ is induced from $<_s$.

or $\sigma_t$ and not describing $t$ explicitly.

In the following, we show that indeed $bs$ is a normal $C\cup R\cup Lin(deq)$-backward simulation from $AbsQ$ to $AbsQ_0$. %

\begin{itemize}
\item[$\langle i \rangle$] $bs[s^{AbsQ}_0] = \{ s^{AbsQ_0}_0 \}$.

\item[\textsc{call-enq}] Let $s \xrightarrow{inv(enq,d,k)}_{AbsQ} s'$ and $t' \in bs[s']$. Either $k \in D_{t'}$ or not. 

First consider the former case. We know that $\ell_{s'}(k) = (d, \texttt{PEND})$ and $k$ is maximal in $s'$. Hence $\sigma_{t'} = \rho \circ \langle d \rangle \circ \pi$ where $\pi$ contains linearization of pending elements in $O_{s'}$. Then, pick $\sigma_t = \rho$. We can find such a $t \in bs[s]$ with $\sigma_t$. Let $(D, \prec)$ be the partial order that is used while constructing $\sigma_{t'}$ from $O_s$ and $<_s$. We can find $(D', \prec')$ for relating $s$ to $t$ such that $D'$ does not contain the values of pending elements that formed $\pi$ suffix of $s_{t'}$ and $d$ coming from linearization of $k \in O_{s'}$. 

One can also see that $t \xrightarrow{\alpha}_{AbsQ_0} t'$ where $\alpha = inv(enq,d,k), lin(enq,d,k), \\lin(enq,d_1,k_1), ..., lin(enq,d_j,k_j)$ such that $\pi = d_1,...,d_j$ and $k_1,...,k_j \in O_{s'}$ are the pending elements that are  linearized to form $\pi$. Note that $\alpha$ obeys the definition of normal backward simulation definition.

For the second case, pick $t$ such that $\sigma_t = \sigma_{t'}$. We can find a $t$ with $\sigma_t$ related to $s$ by $bs$ using the same $(D,\prec)$ partial order that is used while relating $s'$ to $t'$. $\ell_1(\texttt{COMP}(O_s)) \subseteq D$ holds because $\texttt{COMP}(O_s) = \texttt{COMP}(O_{s'})$.

\item[\textsc{call-deq}] Let $s \xrightarrow{inv(deq,d,k)}_{AbsQ} s'$ and $t' \in bs[s']$. Pick $t$ such that it is equal to $t'$ in every field except that $k \notin dom(cp^0_t)$. Then, $t \in bs[s]$ and $t \xrightarrow{inv(deq,d,k)_{AbsQ_0}} t'$.

\item[\textsc{lin-deq1}]  Let $s \xrightarrow{lin(deq,d,k)}_{AbsQ} s'$, $t' \in bs[s']$ and $d \neq \texttt{EMPTY}$. We pick $t$ such that $\sigma_t = \langle d \rangle \circ \sigma_{t'}$. We first show that $t \in bs[s]$. Let $(D, \prec)$ be the partial order that is linearized to obtain $\sigma_{t'}$ and $k' \in O_s$ be the element such that ${\ell_s}_1(k') = d$. We know that $k'$ is minimal in $<_s$ due to the premise of the rule \textsc{lin-deq1}. Hence, we can obtain $(D', \prec')$ such that $D' = D \cup \{{\ell_s}_1(k')\}$ and $\sigma_t$ is a linearization of it.

In addition, $t \xrightarrow{lin(deq,d,k)}_{AbsQ_0} t'$. The action $lin(deq,d,k)$ is enabled in state $t$ since $d$ is the minimum element of $\sigma_t$. Note that the transition relating $t$ to $t'$ obeys the definition of normal forward simulation.

\item[\textsc{lin-deq2}]  Let $s \xrightarrow{lin(deq,\texttt{EMPTY},k)}_{AbsQ} s'$ and $t' \in bs[s']$. We pick $(D, \prec)$ for relating $s$ to $t$ such that $D = \emptyset$. Such a $D$ is a valid choice since all the elements $O_s$ are pending. Then, $\sigma_t = \langle \rangle$ is the only linearization of $(D, \prec)$. Hence, $lin(deq,\texttt{EMPTY},k)$ action is enabled in $AbsQ_0$ and $t \xrightarrow{lin(deq,\texttt{EMPTY},k)}_{AbsQ_0} t'$ holds.

\item[\textsc{ret-enq1}]  Let $s \xrightarrow{ret(enq,k)}_{AbsQ} s'$, $\ell_s(k) = (d,\texttt{PEND})$ and $t' \in bs[s']$. Assume $(D, \prec)$ be the partial order of which linearization is $\sigma_{t'}$. Pick $D'=D$. Then, $\ell_1(\texttt{COMP}(O_s)) \subseteq D \subseteq \ell_1(O_s)$ holds since $\texttt{COMP}(O_s) = \texttt{COMP}(O_{s'}) \setminus \{k\}$ and $k \in \texttt{PEND}(O_s)$. Construct $t \in bs[s]$ such that $\sigma_t = \sigma_{t'}$ is obtained by linearizing the partial order $(D',\prec')$. Then, $t \xrightarrow{ret(enq,k)}_{AbsQ_0} t'$ holds and it is a valid action with respect to normal backward-simulation relation definition. 

\item[\textsc{ret-enq2}]  Let $s \xrightarrow{ret(enq,k)}_{AbsQ} s'$, $k \notin O_s$ and $t' \in bs[s']$. Since $O_s = O_{s'}$ and $\texttt{COMP}(O_s) = \texttt{COMP}(O_{s'})$, we can pick $D' = D$ where $(D, \prec)$ is the strict partial order such that $\sigma_{t'}$ is its linearization. Construct $t \in bs[s]$ such that $\sigma_t = \sigma_{t'}$ is obtained by linearizing the partial order $(D',\prec')$. Then, $t \xrightarrow{ret(enq,k)}_{AbsQ_0} t'$ holds and it is a valid action with respect to normal backward-simulation relation definition. 

\item[\textsc{ret-deq}]  Let $s \xrightarrow{ret(deq,d,k)}_{AbsQ} s'$ and $t' \in bs[s']$. Assume $(D, prec)$ is the partial order of which linearization is $\sigma_{t'}$. Construct $t \in bs[s]$ such that $\sigma_t = \sigma_{t'}$ and $(D, \prec
)$ is the partial order $\sigma_t$ is obtained from. $\texttt{COMP}(O_s) \subseteq D \subseteq O_s$ since $\texttt{COMP}(O_s) = \texttt{COMP}(O_{s'})$ and $O_s = O_{s'}$. Then, $t \xrightarrow{ret(deq,d,k)}_{AbsQ_0} t'$ holds. We have $rv_s(k) = rv^0_t(k)$ since $t \in bs[s]$. Hence the $ret(deq,d,k)$ is enabled in $t$. Moreover, $ret(deq,d,k)$ is a valid transition with respect to the normal backward simulation relation definition.
\end{itemize}
\end{proof}

\begin{lemma} 
$AbsQ_0$ is a refinement of $AbsQ$.
\end{lemma}
\begin{proof}
We define a normal $C\cup R\cup Lin(deq)$-forward simulation $fs$ from $AbsQ_0$ to $AbsQ$ as follows. 
Given $AbsQ_0$ state $t=\tup{\sigma,in^0,rv^0,cp^0}$ and an $AbsQ$ state $s=\tup{O,<,\ell,rv,cp}$ we have that $(t,s)\in fs$ iff the following hold:
\begin{itemize}
	\item the sequence $\sigma$ is a linearization of a partial order $(D,\prec)$ where $D$ contains values labeling elements of $O$ and all the values corresponding to completed enqueues, i.e., $\ell_1({\tt COMP}(O))\subseteq D\subseteq \ell_1(O)$ ordered according to the happens-before order between the enqueues that added them, i.e., $d_1\prec d_2$ if{f} there exists $k_1,k_2$ such that $\ell_1(k_1)=d_1$, $\ell_1(k_2)=d_2$, and $k_1 < k_2$.
	\item every dequeue is at the same control point in both $s$ and $t$, i.e., for every $k$ and $i\in \{1,2,3\}$, $cp(k)=R_i$ iff $cp^0(k)=R_i$,
	\item every enqueue is pending in $s$ whenever it is pending in $t$, i.e., for every $k$, $cp(k)=A_1$ iff $cp^0(k)\in \{A_1,A\}$,
	\item every enqueue is completed in $s$ whenever it is completed in $t$, i.e., for every $k$, $cp(k)=A_2$ iff $cp^0(k) = A_2$,
	\item every pending enqueue which is not linearized or whose value is present in $\sigma$ is a member of $O$, i.e., for every $k$, 
	\begin{align*}
	&k\in O\land \ell(k)=(d,{\tt PEND})\mbox{ iff } \\
	&\hspace{2cm}(cp^0(k)=A_1\land in^0(k)=d)\vee (\exists i.\ \sigma_i =d\land cp^0(k)=A\land in^0(k)=d)
	\end{align*}
	\item every completed enqueue whose value is present in $\sigma$ is a member of $O$, i.e., for every $k$, 
	\begin{align*}
	&k\in O\land \ell(k)=(d,{\tt COMP})\mbox{ iff } \exists i.\ \sigma_i =d\land cp^0(k)=A_2\land in^0(k)=d
	\end{align*}
	\item pending enqueues are maximal in $<$, i.e., for every $k$ and $k'$, $k \not< k'$ if $\ell_2(k)={\tt PEND}$,
	\item the return values fixed at dequeue linearization points are the same, i.e., for every $k$, $rv(k)=rv^0(k)$.
\end{itemize}
In the following, we show that indeed $fs$ is a normal $C\cup R\cup Lin(deq)$-forward simulation from $AbsQ_0$ to $AbsQ$.

\begin{itemize}
\item[$\langle i \rangle$] $fs[s_0^{AbsQ_0}] = \{s_0^{AbsQ}\}$
\item[\textsc{call-enq}] Let $t \xrightarrow{inv(enq,d,k)}_{AbsQ_0} t'$ and $s \in fs[t]$. Then, $inv(enq,d,k)$ is an enabled action in $AbsQ$ since premise of \textsc{call-enq} holds in $t$ and $s \in fs[t]$. Obtain $s'$ such that $s \xrightarrow{inv(enq,d,k)}_{AbsQ} s'$. Note that $s'$ is unique since $AbsQ$ is deterministic with respect to $C \cup R \cup Lin(deq)$. 

Next, we show that $s' \in fs[t']$. Let $(D, \prec)$ be the partial order used while relating $t$ to $s$. Same partial order can be used while relating $\sigma_{s'}$ to $t'$ since $\texttt{COMP}(O_s) = \texttt{COMP}(O_{s'}$, $O_s \subseteq O_{s'}$ and $<_s \subseteq <_{s'}$. The only change we have in control point fields after the actions is that $cp^0_{s'}(k)= A_1$ and $cp_{t'}(k)=A_1$ which satisfies the conditions on $fs$. Moreover $k$ is a maximal pending node in $t'$ as required by the $fs$ conditions. Consequently, $s' \in fs[t']$.

\item[\textsc{call-deq}] Let $t \xrightarrow{inv(deq,k)}_{AbsQ_0} t'$ and $s \in fs[t]$. Then, $inv(deq,k)$ is an enabled action in $AbsQ$ since premise of \textsc{call-deq} holds in $t$ and $s \in fs[t]$. Obtain $s'$ such that $s \xrightarrow{inv(deq,k)}_{AbsQ} s'$. Note that $s'$ is unique since $AbsQ$ is deterministic with respect to $C \cup R \cup Lin(deq)$. 

Next, we show that $s' \in fs[t']$. Since $\sigma_s =\sigma_{s'}$, $O_s = O_{s'}$ and $\texttt{COMP}(O_s) = \texttt{COMP}(O_{s'}$, we can pick same $(D,\prec)$  partial order in $s'$ and show that $\sigma_{t'}$ is a linearization of it. The only change in control points after the transitions is that $cp^0_{t'}(k) = cp_{s'}(k) = R_1$ which does not violate any condition in $fs$. Consequently, $s' \in fs[t']$. 

\item[\textsc{lin-enq}] Let $t \xrightarrow{lin(enq,d,k)}_{AbsQ_0} t'$ and $s \in fs[t]$. Then, pick $s' = s$ such that $s \xrightarrow{\epsilon}_{AbsQ} s'$. Note that $\epsilon$ is a valid transition with respect to the normal forward simulation relation definition. We show that $s \in fs[t']$. If $(D, \prec)$ is the partial order in $s$ of which one linearization is $\sigma_t$, we pick $D' = D \cup \{k\} \subseteq O_s$. $(D', \prec')$ can be linearized to $\sigma_{t'}$ since $k$ is a maximal pending node and can be linearized at the end. Moreover, the only change in control point $cp^0_{t'}(k) = A$ which does not violate the $fs$ conditions.

\item[\textsc{lin-deq1}] Let $t \xrightarrow{lin(deq,d,k)}_{AbsQ_0} t'$, $d \neq \texttt{EMPTY}$ and $s \in fs[t]$. Then, $lin(deq,d,k)$ is an enabled action in $AbsQ$. There must exist $d \in D \subseteq \ell_1(O_s)$ such that ${\ell_s}_1(k') = d$ and $k'$ is minimal in $D$ (since $d$ is linearized as the minimum element in $\sigma_t$ according to premise of \textsc{lin-deq1} of $AbsQ$). Obtain $s'$ such that $s \xrightarrow{lin(deq,d,k)}_{AbsQ} s'$. Note that $s'$ is unique since $AbsQ$ is deterministic with respect to $C \cup R \cup Lin(deq)$. 

Next, we show that $s' \in fs[t']$. Let $(D, \prec)$ be the partial order used while relating $t$ to $s$ such that $\sigma_t$ is a linearization of this partial order. Since we have shown that $k'$ is minimal in that partial order, $\sigma_{t'}$ is a linearization of $(D', \prec')$ where $D' = D \setminus \{\ell_1(k')\}$. Note that $\ell_1(\texttt{COMP}(O_{s'})) \subseteq D' \subseteq \ell_1(O_{s'})$ holds. The only change in control points is that $cp^0_{t'}(k) = cp_{s'}(k) = R_2$ which does not violate the conditions for relating $t'$ to $s'$. Note that the fifth condition of $fs$ still holds for $k'$ while relating $t'$ to $s'$. After transitions $rv^0_{t'}(k) = rv_{s'}(k) = d$ and the last condition on $fs$ is preserved.

\item[\textsc{lin-deq2}] Let $t \xrightarrow{lin(deq,\texttt{EMPTY},k)}_{AbsQ_0} t'$ and $s \in fs[t]$. Then, $lin(deq,d,k)$ is an enabled action in $AbsQ$. If $\texttt{COMP}(O_t) \neq \emptyset$, then $D$ use for linearization cannot be $\emptyset$ $\sigma_t = \langle\rangle$ cannot be a linearization of $(D, \prec)$. Obtain $s'$ such that $s \xrightarrow{lin(deq,\texttt{EMPTY},k)}_{AbsQ} s'$. Note that $s'$ is unique since $AbsQ$ is deterministic with respect to $C \cup R \cup Lin(deq)$. 

Next, we show that $s' \in fs[t']$. Let $(D, \prec)$ be the partial order used while relating $t$ to $s$ such that $\sigma_t = \langle \rangle $ is a linearization of this partial order. We can use the same $(D, \prec)$ for relating $t'$ to $s'$ because $\sigma$ field is the same for both $s$, $s'$; and $O$, $<$, $\ell$ fields are same for both $t$ and $t'$. The only change in control points is that $cp^0_{t'}(k) = cp_{s'}(k) = R_2$ which does not violate the conditions for relating $t'$ to $s'$.  After transitions $rv^0_{t'}(k) = rv_{s'}(k) = d$ and the last condition on $fs$ is preserved.

\item[\textsc{ret-enq}] Let $t \xrightarrow{ret(enq,k)}_{AbsQ_0} t'$ and $s \in fs[t]$. Then, there are two cases assuming data independence: (i) $in^0_t(k)=d$ and $\exists i. \sigma_t(i)= d$ (ii) or not.   

First, consider the former case. Then, \textsc{ret-enq1} rule of $AbSQ$ is applicable. Its precondition holds since fifth condition of $fs$ holds while relating $t$ to $s$. Apply this rule ($ret(enq,k)$) to obtain $s'$. Note that $s'$ is unique since $AbsQ$ is deterministic with respect to $C \cup R \cup Lin(deq)$ and it is a valid action according to the normal forward-simulation relation definition.

Next, we show that $s' \in fs[t']$. Since $\exists i. \sigma_t(i) = d$, we know that $d \in D$ where $(D, \prec)$ is the partial order satisfying first condition of $fs$ while relating $t$ to $s$, and $k \in O_s$ takes part in the linearization i.e., ${\ell_t}_1(k) \in D$. We can use the same partial order $(D,\prec)$ for relating $t'$ to $s'$ such that it satisfies the first condition of $fs$. The only change in control points is that $cp^0_{t'}(k) = cp_{s'}(k) = A_2$ which does not violate the conditions for relating $t'$ to $s'$. Note that the sixth condition of $fs$ also continue to hold for $k$ for the post-states.

Second, consider the latter case:  $in^0_t(k)=d$, but $\forall i. \sigma_t(i)\neq d$. Since $(t,s) \in fs$, $k \notin O_s$ by the fifth and sixth conditions. Hence, the pre-condition of \textsc{ret-enq2} is satisfied by $t$. Apply this rule ($ret(enq,k)$) to obtain $s'$. Note that $s'$ is unique since $AbsQ$ is deterministic with respect to $C \cup R \cup Lin(deq)$ and it is a valid action according to the normal forward-simulation relation definition.

Next, we show that $s' \in fs[t']$. For satisfying the first condition, one can use the same $(D, \prec)$ partial order that is used for relating pre-states since $\sigma$ fields of $t$, $t'$ and $O$, $<$, $\ell$ fields of $s$ and $s'$ are the same.  The only change in control points is that $cp^0_{t'}(k) = cp_{s'}(k) = A_2$ which does not violate the conditions for relating $t'$ to $s'$.

\item[\textsc{ret-deq}] Let $t \xrightarrow{ret(deq,d,k)}_{AbsQ_0} t'$ and $s \in fs[t]$. Then, $ret(deq,d,k)$ is an enabled action in $AbsQ$ due to premise \textsc{ret-deq} of $AbsQ_0$ and the last condition on $fs$ (since $(t,s) \in fs$). Obtain $s'$ such that $s \xrightarrow{ret(deq,d,k)}_{AbsQ} s'$. Note that $s'$ is unique since $AbsQ$ is deterministic with respect to $C \cup R \cup Lin(deq)$.

We see that $s' \in fs[t']$. Pre-states are equal to the post-states with the only exception in the control points such that $cp^0_{t'}(k) = cp_{s'}(k) = R_3$. All the conditions except the third one continues to hold in the post states since they hold in the pre-states. The third rule regarding the control points of dequeues also continue to the hold since changes in the control point of $k$ does not violate it.
\end{itemize}
\end{proof} %
\section{Proof of Theorem~\ref{th:absImplStack}}\label{app:absImplStack}

We show that $AbsS$ and $AbsS_0$ refine each other. The standard reference implementation $AbsS_0$ is defined exactly as the one for queues, $AbsQ_0$, except that pop linearization points extract values from the beginning of the sequence stored in the state.

\begin{figure} [t]
{\scriptsize
  \centering
  \begin{mathpar}
    \inferrule[call-push]{
      k\not\in dom(cp^0) \\ 
      d\neq {\tt EMPTY}
    }{
      \sigma,in^0,rv^0,cp^0
      \xrightarrow{inv(push,d,k)}
      \sigma,in^0[k\mapsto d], rv^0,cp^0[k\mapsto A_1]
    }\hspace{5mm}
    \inferrule[lin-push]{
      cp^0(k)=A_1
    }{
      \sigma,in^0,rv^0,cp^0
      \xrightarrow{lin(push,d,k)}
      d\cdot\sigma,in^0,rv^0,cp^0[k\mapsto A]
    }\hspace{5mm}
    
        \inferrule[ret-push]{
      cp^0(k)=A
    }{
      \sigma,in^0,rv^0,cp^0
      \xrightarrow{ret(push,k)}
      \sigma,in^0,rv^0,cp^0[k\mapsto A_2]
    }\hspace{5mm}

    \inferrule[call-pop]{
      k\not\in dom(cp^0) \\ 
    }{
      \sigma,in^0,rv^0,cp^0
      \xrightarrow{inv(pop,k)}
      \sigma,in^0,rv^0,cp^0[k\mapsto R_1]
    }\hspace{5mm}
        \inferrule[lin-pop1]{
      cp^0(k)=R_1 \\
      \sigma = d\cdot \sigma' 
    }{
      \sigma,in^0,rv^0,cp^0
      \xrightarrow{lin(pop,d,k)}
      \sigma',in^0,rv^0[k\mapsto d],cp^0[k\mapsto R_2]
    }\hspace{5mm}

        \inferrule[lin-pop2]{
      cp^0(k)=R_1 \\
      \sigma = \epsilon
    }{
      \sigma,in^0,rv^0,cp^0
      \xrightarrow{lin(pop,{\tt EMPTY},k)}
      \sigma,in^0,rv^0[k\mapsto {\tt EMPTY}],cp^0[k\mapsto R_2]
    }\hspace{5mm}
    \inferrule[ret-pop]{
      cp^0(k)=R_2 \\
      rv^0(k) = d
    }{
      \sigma,in^0,rv^0,cp^0
      \xrightarrow{ret(pop,d,k)}
      \sigma,in^0,rv^0,cp^0[k\mapsto R_3]
    }\hspace{5mm}
          \end{mathpar}
  }
 \vspace{-4mm}
  \caption{The transition relation of $AbsS_0$. 
  }
  \label{fig:transitions:AbsS_0}
\vspace{-2mm}
\end{figure}

Thus, the states of $AbsS_0$ are tuples $\tup{\sigma,in^0,rv^0,cp^0}$ where $\sigma\in\<Vals>^*$ is a sequence of values, $in^0:\<Ops> ~> \<Vals>$ records the input value of a push, $rv^0:\<Ops> ~> \<Vals>$ records the return value of a pop fixed at its linearization point ($~>$ denotes a partial function), and $cp^0:\<Ops> ~> \{A_1,A,A_2,R_1,R_2,R_3\}$ records the control point of every push ($A_1, A,A_2$) or pop operation ($R_1,R_2,R_3$).
All the components are $\emptyset$ in the initial state, and the transition relation $->$ is defined in Fig.~\ref{fig:transitions:AbsS_0}. The alphabet of $AbsS$ contains call/return actions and push/pop linearization points.

To prove that $AbsS$ is a refinement of $AbsS_0$ we define a normal $C\cup R$-backward simulation (i.e, a backward simulation as in Definition~\ref{def:back_app}) from $AbsS$ to $AbsS_0$. The reverse is shown using a normal $C\cup R$-forward simulation (i.e, a forward simulation as in Definition~\ref{def:for_app}).

\begin{lemma} 
$AbsS$ is a refinement of $AbsS_0$.
\end{lemma}
\begin{proof}
We define a normal $C\cup R$-backward simulation $bs$ from $AbsS$ to $AbsS_0$ as follows. Given an $AbsS$ state $s=\tup{O,<,\ell,rv,cp}$ and an $AbsS_0$ state $t=\tup{\sigma,in^0,rv^0,cp^0}$ we have that $(s,t)\in bs$ iff the following hold:
\begin{itemize}
	\item if a pop has committed or respectively, it has returned in $s$, then it had been linearized or respectively, it has returned in $t$, i.e., for every $k$, if $cp(k)\in\{R_2,R_3\}$ then $cp^0(k)=cp(k)$,
	\item a push is completed in $s$ whenever the same is true in $t$, i.e., for every $k$, $cp(k)=A_2$ iff $cp^0(k)=A_2$,
	\item a push is pending in $s$ iff either it is a non-linearized pending push in $t$ or its linearization point has been executed, i.e., for every $k$, $cp(k)=A_1$ iff $cp^0(k)=A_1$ and $in^0(k)$ doesn't occur in $\sigma$, or $cp(k)=A$,
	\item if a pop didn't commit in $s$ then it is pending in $t$ and it may have been linearized, i.e., for every $k$, if $cp(k)=R_1$ then $cp^0(k)\in \{R_1,R_2\}$,
	\item there exists a partial injective function $g: \{k: cp(k)=R_1\} ~> O$ which associates uncommitted pops to pushes in $O$ such that:
	\begin{itemize}
		\item for every $k$, $k\in dom(g)$ iff $g(k)\in be(k)\cup ov(k)$
		\item  the sequence $\sigma$ is the mirror of a linearization of a partial order $(D,\prec)$ where $D$ contains values labeling elements of $O$ except for those in the range of $g$, and all the values corresponding to completed pushes which are not in the range of $g$, i.e., $\ell_1({\tt COMP}(O)\setminus range(g))\subseteq D\subseteq \ell_1(O\setminus range(g))$ ordered according to the happens-before order between the pushes that added them, i.e., $d_1\prec d_2$ if{f} there exists $k_1,k_2$ such that $\ell_1(k_1)=d_1$, $\ell_1(k_2)=d_2$, and $k_1 < k_2$
		\item every pop in the domain of $g$ has been linearized, i.e., for every $k$, $k\in dom(g)$ implies $cp^0(k)=R_2$,
		\item every pop which is not in the domain of $g$ hasn't been linearized, i.e., for every $k$, $k\not\in dom(g)$ implies $cp^0(k)=R_1$,
		\item every push in the range of $g$ has been linearized, i.e., for every $k$, $k\in range(g)$ implies $cp^0(k)=A$,
		\item a pending enqueue from $O$ has been linearized when its value is contained in $\sigma$, i.e., for every $k$, if $\ell_1(k)\in D$ and $\ell_2(k)={\tt PEND}$, then $cp^0(k)=A$.
	\end{itemize}
	\item the return values fixed at pop commit points are the same, i.e., for every $k$, if $rv(k)$ is defined, then $rv(k)=rv^0(k)$,
	\item every pending push has the same input value in both $s$ and $t$, i.e., for every $k$, $\ell_1(k)=in^0(k)$,
\end{itemize}

In the following, we show that indeed $bs$ is a normal $C\cup R$-backward simulation from $AbsS$ to $AbsS_0$:
\begin{itemize}
	\item Let $s\xrightarrow{inv(push,d,k)} s'$ be a transition in $AbsS$ and $(s',t')\in bs$. We consider two cases depending on whether the value $d$ occurs on a position $i$ in the sequence $\sigma$ of $t'$ or not. If it occurs, let $t$ be a $AbsS_0$ state where essentially, the component $\sigma$ is the prefix of the sequence $\sigma$ of $t'$ that contains the first $i-1$ positions (except for some set of pushes that will be defined hereafter, all operations are at the same control point). Let $@t$ be the following $AbsS_0$ trace:
	\begin{align*}
	@t=inv(push,d,k),\ lin(push,d,k),\ lin(push,d_{i+1},k_{i+1}),\ldots,lin(push,d_{n-1},k_{n-1})
	\end{align*}
	where $d_j$ is the value on position $j$ in the sequence $\sigma$ of $t'$ and $n$ is the length of this sequence (we assume that positions are indexed starting from $0$). Let $k_i=k$.
	For every $k_j$ with $i\leq j\leq n-1$, we must have that $cp(k_j)=A$ in $t'$. We take $cp(k_j)=A_1$ in $t$ for every $j\geq i+1$ and $cp$ undefined for $k_i$.
	We have that $t\xrightarrow{@t} t'$ is a valid sequence of transitions of $AbsS_0$ and $(s,t)\in bs$ (the latter can be proved by taking the same function $g$ used in establishing that $(s',t')\in bs$).
	Now, assume that the value $d$ is not in the sequence $\sigma$ of $t'$.  We consider an $AbsS_0$ state $t$ where the component $\sigma$ is the same as the one in $t'$. There are two sub-cases depending on whether there exists a pending pop $k'$ such that $g(k')=k$ when establishing that $(s',t')\in bs$. If it exists, the operations are at the same control point in both $t$ and $t'$ except for the push $k$ for which $cp(k)$ is undefined in $t$, and the the pop $k'$ for which we take $cp(k')=R_0$ in $t$. We have that 
	\begin{align*}
	t\xrightarrow{inv(push,d,k),\ lin(push,d,k)\ lin(pop,d,k')} t'
	\end{align*}
	in $AbsS_0$. If there exists no such pop $k'$, it can be easily seen that there exists $t$ such that $t\xrightarrow{inv(push,d,k)} t'$ and $(s,t)\in bs$.	

	\item Let $s\xrightarrow{inv(pop,k)} s'$ be a transition in $AbsS$ and $(s',t')\in bs$. We consider two cases depending on whether in the function $g$ used to relate $s'$ to $t'$ we have that $k\in dom(g)$. In other words, either the newly invoked pop operation $k$ did not linearize yet ($k\not\in dom(g)$) or it linearizes and removes an element inserted by a linearized push ($k\in dom(g)$). The second case also splits into two sub-cases: The value removed by pop $k$ is inserted by a push $k'=g(k)$ that is still pending or the push has returned. We will look at all three cases separately. 
	The easiest one is the first case. There exists some $t$ where essentially the component $\sigma$ is the same as the one in $t'$, such that $t\xrightarrow{inv(pop,k)} t'$ and $(s,t)\in bs$. 
	
	For the first sub-case of the second case, we take an $AbsS_0$ state $t$ where $\sigma_t = d\cdot\sigma_{t'}$ (we use $\sigma_t$ to denote the component $\sigma$ in $t$). It must happen that $cp^0_{t'}(k)=R_2$. The operations are at the same control point in both $t$ and $t'$, except for $k$ in which case $cp^0$ is undefined. We have that $t\xrightarrow{inv(pop,k),\ lin(pop,d,k')} t'$ and $(s,t)\in bs$. The latter holds because essentially, $k'$ is a maximal node in $s'$ (since it is pending). 

	For the second sub-case, we define an $AbsS_0$ state $t$ where the sequence $\sigma$ is the minimal prefix of $\sigma_{t'}$ that includes the value $d$ added by $k'$. Let $i$ be the index of this value in $\sigma_{t'}$ and $k_{j}$ with $i<j$ the identifiers of the pushes that added the values following $d$ in $\sigma_{t'}$.  Let $@t$ be the following $AbsS_0$ trace:
	\begin{align*}
	@t=inv(pop,k),\ lin(pop,d,k),\ lin(push,d_{i+1},k_{i+1}),\ldots,lin(push,d_{n-1},k_{n-1})
	\end{align*}
	where $d_j$ is the value on position $j$ in the sequence $\sigma_{t'}$ and $n$ is the length of this sequence. We have that $t\xrightarrow{@t} t'$ is a valid sequence of transitions of $AbsS_0$ and $(s,t)\in bs$. The latter relies on the fact that $k'$ is a greatest completed push in $s$ and all pushes $k_j$ with $j>i$ are pending in $s$.
	
	\item Let $s\xrightarrow{com(pop,d,k)} s'$ be a transition in $AbsS$ and $(s',t')\in bs$. When this transition results in removing a greatest completed push or a pending push in $s$, then there exists an $AbsS_0$ state $t$ such that $t\xrightarrow{@t} t'$ is a valid sequence of $AbsS_0$ transitions and $(s,t)\in bs$, for some $t$ and $@t$ defined as in the second case of $inv(pop,k)$. When it removes a completed push which is followed by other completed pushes (in the happens-before in $s$), then we pick $t=t'$. We have that $t\xrightarrow{\epsilon} t'$ and $(s,t)\in bs$ (for the latter we must choose a function $g$ such that $g(k)=k'$ where $k'$ is the push removed by the $AbsS$ transition.
	
	\item Let $s\xrightarrow{ret(push,k)} s'$ be a transition in $AbsS$ and $(s',t')\in bs$. We consider two cases depending on whether the happens-before in $s$ contains push $k$. If it contains push $k$, there are two sub-cases: (1) if its input is present in $\sigma_{t'}$ then there exists an $AbsS_0$ state $t$ such that $t\xrightarrow{ret(push,k)} t'$ is a valid sequence of $AbsS_0$ transitions and $(s,t)\in bs$, and (2) otherwise, we take a state $t$ where essentially, $\sigma_t = d\cdot\sigma_{t'}$ for which we have that $t\xrightarrow{lin(pop,d,k),\ ret(push,k)} t'$ and $(s,t)\in bs$. If the happens-before in $s$ doesn't contain push $k$, then there exists an $AbsS_0$ state $t$ such that $t\xrightarrow{ret(push,k)} t'$.
	
	\item The case of pop returns $ret(pop,k)$ is trivial. Such transitions are simulated by $ret(pop,k)$ transitions of $AbsS_0$.
\end{itemize}

\end{proof}

\begin{lemma} 
$AbsS_0$ is a refinement of $AbsS$.
\end{lemma}
\begin{proof}
We define a normal $C\cup R$-forward simulation $fs$ from $AbsS_0$ to $AbsS$ as follows. Given an $AbsS_0$ state $t=\tup{\sigma,in^0,rv^0,cp^0}$ and an $AbsS$ state $s=\tup{O,<,\ell,rv,cp}$ we have that $(t,s)\in fs$ iff the following hold:
\begin{enumerate}
	\item every pop is at the same control point in both $t$ and $s$, i.e., for every $k$ and $i\in \{1,2,3\}$, $cp^0(k)=R_i$ iff $cp(k)=R_i$,
	\item a push has been invoked in $t$ whenever it has been invoked in $s$, i.e., for every $k$, $cp^0(k)=A_1$ iff $cp(k)=A_1$,
	\item a push which is linearized in $t$ has been invoked in $s$, i.e., for every $k$, if $cp^0(k)=A$ then $cp(k)=A_0$,
	\item a push is completed in $t$ iff the same holds in $s$, i.e., for every $k$, $cp^0(k)=A_2$ iff $cp(k)=A_2$,
	\item the pair $(O,\ell)$ in $s$ satisfies the following:
	\begin{itemize}
		\item for every $k$, if $in^0(k)=d$, $cp^0(k)\in \{A_1,A\}$, and $d$ occurs in $\sigma$, then $k\in O$ and $\ell(k)=(d,{\tt PEND})$,
		\item for every $k$, if $in^0(k)=d$, $cp^0(k) = A_2$, and $d$ occurs in $\sigma$, then $k\in O$ and $\ell(k)=(d,{\tt COMP})$,
	\end{itemize}
	\item every pending push in $O$ is overlapping with every non-linearized pop, i.e., for every $k$, if $cp^0(k)=R_1$ then $\{k': k'\in O\land \ell_2(k')={\tt PEND}\}\subseteq ov(k)$.
	\item every completed push is either overlapping or was the greatest completed push before a non-linearized pop started, i.e., for every $k$, if $cp^0(k)=R_1$, then ${\tt COMP}(O)\subseteq ov(k)\cup be(k)$,
	\item for every push that overlaps with a pop $k$ or was maximal in $<$ when $k$ started, its successors are overlapping with $k$, i.e., $k_1\in be(k)\cup ov(k)$ and $k_1 < k_2$ implies $k_2 \in ov(k)$ for each $k, k_1, k_2$
	\item predecessors of pushes in $be(k)$ for a given pop $k$ are neither overlapping with $k$ nor in $be(k)$, i.e., $k_1 < k_2$ and $k_2\in be(k)$ implies $k_1\not\in ov(k)\cup be(k)$ for each $k,k_1,k_2$
	\item pending pushes are maximal in $<$, for every $k$ and $k'$, $k \not< k'$ if $\ell_2(k)={\tt PEND}$,
	\item\label{item:stack_fs_lin} the sequence $\sigma$ is the mirror of a linearization of a partial order $(D,\prec)$ where $D$ contains values labeling elements of $O$ and all the values corresponding to completed pushes, i.e., $\ell_1({\tt COMP}(O))\subseteq D\subseteq \ell_1(O)$ ordered according to the happens-before order between the pushes that added them, i.e., $d_1\prec d_2$ if{f} there exists $k_1,k_2$ such that $\ell_1(k_1)=d_1$, $\ell_1(k_2)=d_2$, and $k_1 < k_2$.
	\item the return values fixed at pop linearization/commit points are the same, i.e., for every $k$, $rv(k)=rv^0(k)$,
	\item every pending push has the same input value in both $s$ and $t$, i.e., for every $k$, $\ell_1(k)=in^0(k)$,
\end{enumerate}

In the following, we show that indeed $fs$ is a normal $C\cup R$-backward simulation from $AbsS_0$ to $AbsS$:
\begin{itemize}
	\item Let $t\xrightarrow{inv(push,d,k)} t'$ be a transition in $AbsS_0$ and $(t,s)\in fs$. We have that $(t',s')\in fs$ where $s\xrightarrow{inv(push,d,k)} s'$ (recall that $AbsS$ is deterministic). Since the push $k$ is non-linearized in $t'$, the component $\sigma$ of both $t$ and $t'$ are the same and $\ell_2(k)={\tt PEND}$ in $s'$. Then, the component $\sigma$ in $Abs_0$ states related by $fs$ to $s'$ is allowed to exclude values added by pushes in $s'$ which are labeled as pending. The effect of $inv(push,d,k)$ in $AbsS$ implies that $k$ overlaps with all pending pops.
	\item Let $t\xrightarrow{inv(pop,k)} t'$ be a transition in $AbsS_0$ and $(t,s)\in fs$. We have that $(t',s')\in fs$ where $s\xrightarrow{inv(pop,k)} s'$. The only difference between $s$ and $s'$ is that the components $be(k)$ and $ov(k)$ in $s'$ contain the greatest completed pushes in $s$ and the pending pushes in $s$, respectively (these components were undefined in $s$). The relation $fs$ doesn't exclude this particular choice for $be(k)$ and $ov(k)$ when applied to $t'$ and $s'$.
	\item Let $t\xrightarrow{lin(push,d,k)} t'$ be a transition in $Abs_0$ and $(t,s)\in fs$. We have that $(t',s)\in fs$, i.e., the $AbsS_0$ transition is simulated by an empty sequence of $AbsS$ transitions, because essentially the component $\sigma$ of $t'$ still corresponds to a linearization of the pushes in $s$ according to item \ref{item:stack_fs_lin} in the definition of $fs$. The sequence $\sigma$ in $t'$ contains the value added by the push $k$ at the end, but this is allowed by $fs$ since $k$ is labeled as pending in $s$.
	\item Let $t\xrightarrow{lin(pop,d,k)} t'$ be a transition in $AbsS_0$ and $(t,s)\in fs$. We have that $(t',s')\in fs$ where $s\xrightarrow{com(pop,d,k)} s'$. The transition labeled by $com(pop,d,k)$ is enabled in $AbsS_0$ because $d$ was the first value in the sequence $\sigma$ of $t$. Indeed, this implies that $d$ was added by a push $k'$ which is maximal in the happens-before stored in $s$. This clearly implies that $k'\in be(k)\cup ov(k)$. In addition, the sequence $\sigma$ in $t'$ does correspond to a linearization of the pushes in $s'$ (which don't contain $k$ anymore) because $\sigma$ in $t$ had this property with respect to $s$ and $\sigma$ in $t'$ is obtained by deleting the first value in the sequence $\sigma$ of $t$. 
	\item Let $t\xrightarrow{ret(push,k)} t'$ be a transition in $AbsS_0$ and $(t,s)\in fs$. We have that $(t',s')\in fs$ where $s\xrightarrow{ret(push,k)} s'$. There are two cases depending on whether the value added by $k$ is still present in the sequence $\sigma$ of $t$. If it is not, then the push $k$ doesn't occur in the happens-before from $s$, and the only effect of these two transitions is changing the control point of $k$. Therefore, $(t',s')\in fs$ clearly holds. When this value is still present, the effect of $ret(push,k)$ in $AbsS$ is changing the flag of push $k$ from ${\tt PEND}$ to ${\tt COMP}$. Since the order between pushes doesn't change, we have that $(t',s')\in fs$.
	\item Let $t\xrightarrow{ret(pop,d,k)} t'$ be a transition in $AbsS_0$ and $(t,s)\in fs$. We have that $(t',s')\in fs$ where $s\xrightarrow{ret(pop,d,k)} s'$. This case is obvious, the only change between $s$ and $s'$ being the control point of $k$.
\end{itemize}

\end{proof}
\section{Proving the correctness of $TSS$}\label{app:tss}

\begin{figure}[t]
\centering
\includegraphics[width=7cm]{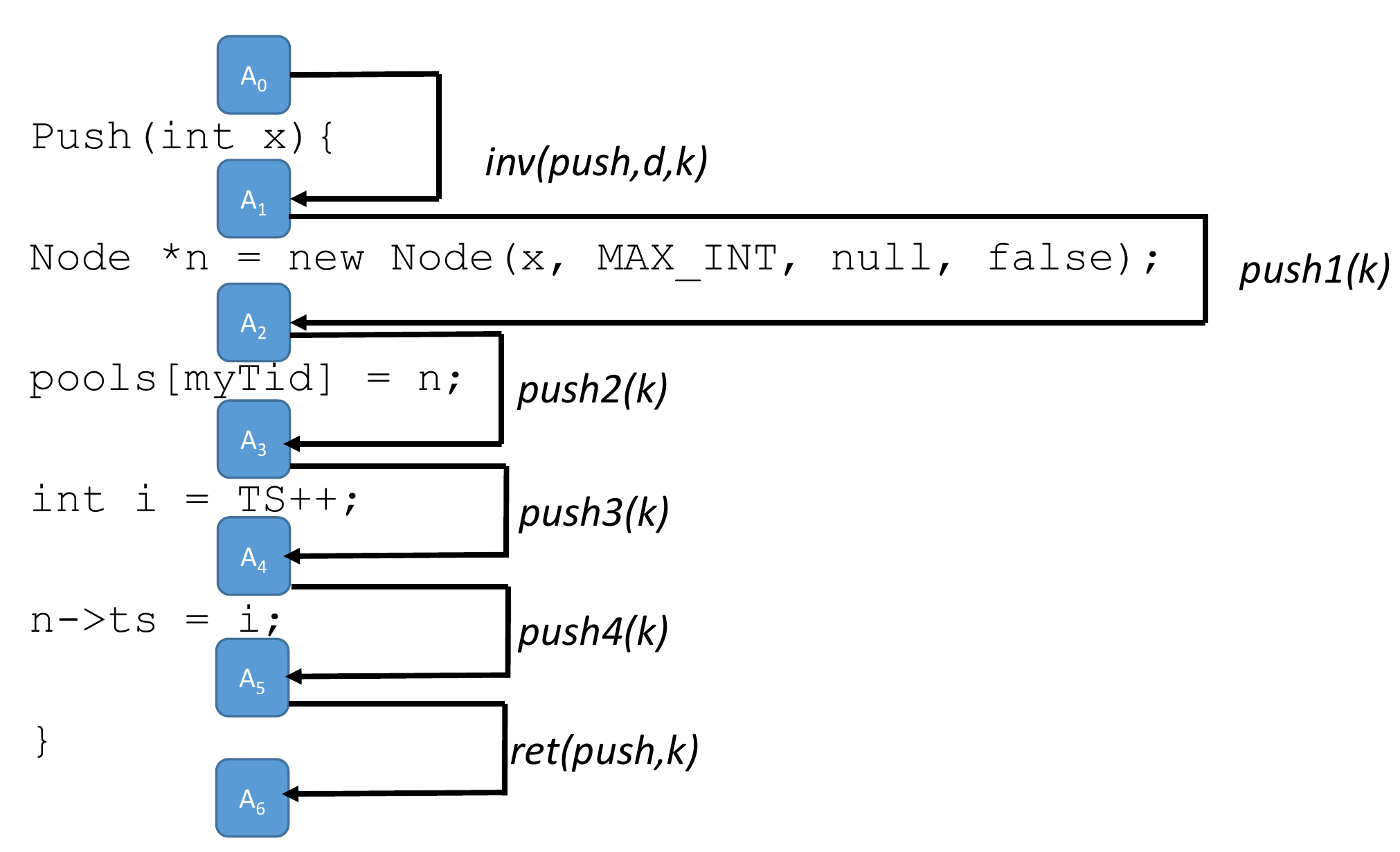}
\includegraphics[width=16cm]{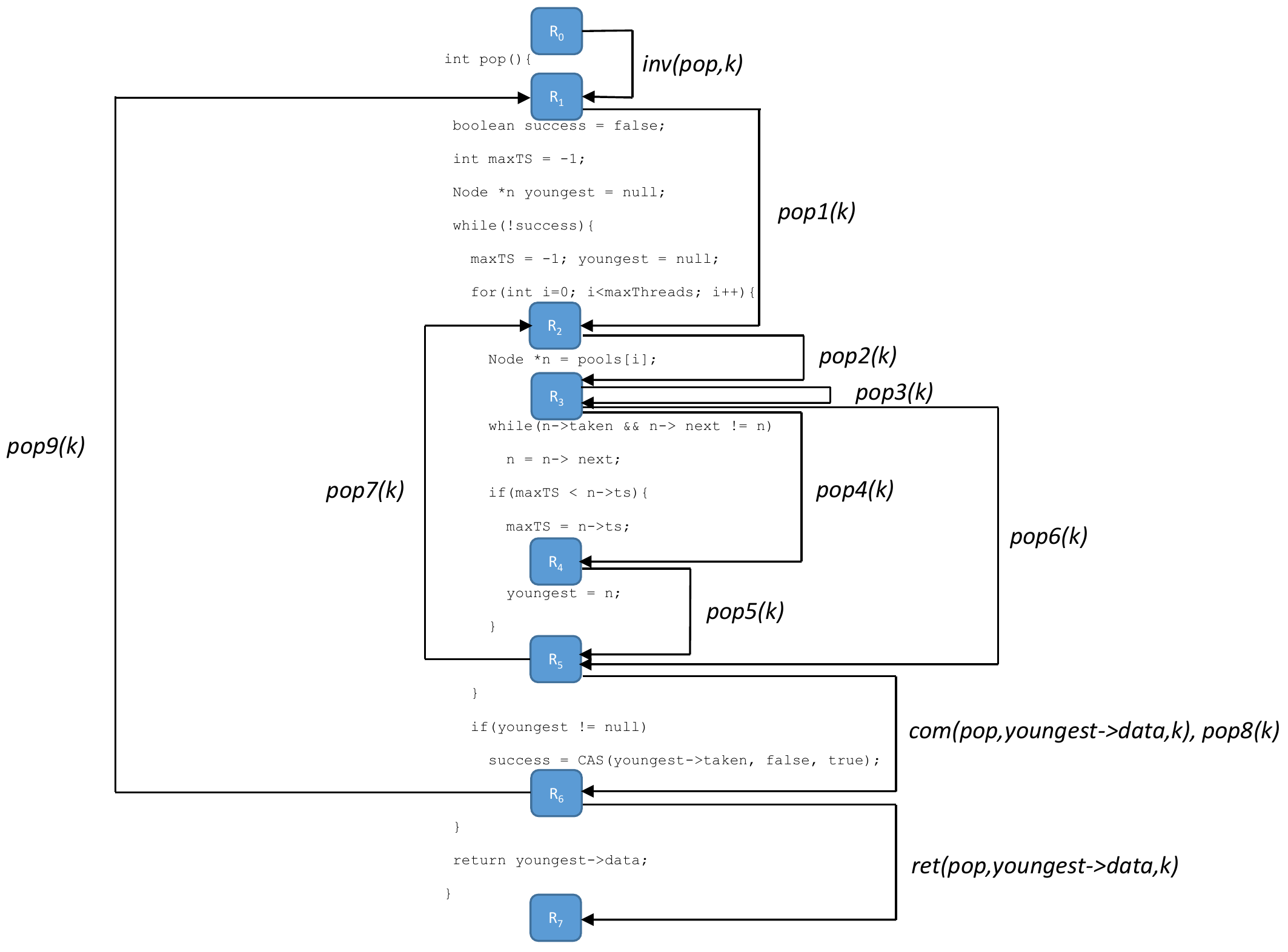}
\vspace{-8mm}
\caption{The flow diagram for the pop and push methods of the Time-Stamped Stack algorithm. The blue points show the control points roughly and the arrows show the possible transitions.}
\label{fig:tssFlow}
\end{figure}

The LTS corresponding to the description of $TSS$ given in Fig.~\ref{fig:TimeStamped} is defined as usual. The control points and transition labels we use in the following proof are pictured in Fig.~\ref{fig:tssFlow}. To simplify the proof, we take the initializations of some local variables together as atomic.

States of the TS-Stack contains the global variables and local variables as fields. Global variables are just elements of their domains and local variables are maps from operation identifiers to their domains. We say $i_q(k)$ for referencing the value of local variable $i$ of operation $k$ in state $q$. There is only one special local variable called $myTID$. Its value is unique to each pending operation in a state i.e., concurrent operations cannot have the same $myTID$ value. TS-Stack states also contains sets $O_a, O_r \in \mathbb{O}$ which are operation identifier sets of push and pops respectively, and the control point function $cp$ which is a map from operation identifiers to the control points set that are presented in the flow diagram Figure ~\ref{fig:tssFlow}. Transition relation of the TS-STack is presented in Figure~\ref{fig:transitions:TSSPush} (push rules) and Figure~\ref{fig:transitions:TSSPop} (pop rules).
Next, we show that the linearizability of TS Stack.

\begin{figure} [t]
{\scriptsize
  \centering
  \begin{mathpar}
    \inferrule[call-push]{
      k\not\in dom(cp) \\ 
      d\neq {\tt null}
    }{
      ..., O_a, x, cp,...
      \xrightarrow{inv(push,d,k)} 
     ..., O_a\cup\{k\}, x[k \mapsto d], cp[k \mapsto A_1], ...
    }\hspace{5mm}

    \inferrule[push1]{
      cp(k) = A_1 \\ 
      *n' = (x(k),\texttt{MAX\_INT}, \texttt{null}, \texttt{false})
    }{
      ..., n, cp,...
      \xrightarrow{push1(k)} 
     ..., n[k \mapsto n'], cp[k \mapsto A_2], ...
    }\hspace{5mm}

  \inferrule[push2]{
      cp(k) = A_2 
    }{
      ..., pools, cp,...
      \xrightarrow{push2(k)} 
     ..., pools[myTid(k) \mapsto n(k)], cp[k \mapsto A_3], ...
    }\hspace{5mm}
    
      \inferrule[push3]{
      cp(k) = A_3 
    }{
      ..., i, TS, cp,...
      \xrightarrow{push3(k)} 
     ..., i[k \mapsto TS], TS+1, cp[k \mapsto A_4], ...
    }\hspace{5mm}
    
   \inferrule[push4]{
      cp(k) = A_4 \\
      n'(k) = n(k)[ts \mapsto i(k)] \\
      \forall k'. cp(k') = A_6 \implies i(k') < i(k)
    }{
      ..., n, cp,...
      \xrightarrow{push4(k)} 
     ..., n[k \mapsto n'(k)], cp[k \mapsto A_5], ...
    }\hspace{5mm}
    
   \inferrule[ret-push]{
      cp(k) = A_5 
    }{
      ...,cp,...
      \xrightarrow{ret(push,k)} 
     ...,cp[k \mapsto A_6], ...
    }\hspace{5mm}
    
      \end{mathpar}
  }
 \vspace{-5mm}
  \caption{The push derivation rules of $TSS$. We only mention the state components that are modified. Unmentioned state components have the names in the algorithm in the prestate. $*n = (a,b,c,d)$ is shorthand for $n->data = a$, $n->ts = b$, ... $n' = n[ts \mapsto expr]$ is short for $n'->ts = expr$ and all the other fields of $n$ and $n'$ are the same.
  }
  \label{fig:transitions:TSSPush}
\vspace{-2mm}
\end{figure}

\begin{figure} [t]
{\scriptsize
  \centering
  \begin{mathpar}
    \inferrule[call-pop]{
      k\not\in dom(cp) 
    }{
      ..., O_r, cp,...
      \xrightarrow{inv(pop,k)} 
     ..., O_r\cup\{k\}, cp[k \mapsto R_1], ...
    }\hspace{5mm}
    
    \inferrule[pop1]{
      cp(k) = R_1 \\
      maxThreads > 0
    }{
      ..., suc, ygst, mTS, i, cp
      \xrightarrow{pop1(k)} 
     ..., suc[k \mapsto \texttt{false}], ygst[k \mapsto \texttt{null}], mTS[k \mapsto -1], i[k \mapsto 0],  cp[k \mapsto R_2]
    }\hspace{5mm}
    
    \inferrule[pop2]{
      cp(k) = R_2 \\
      0 \leq i(k) < \texttt{maxThreads}
    }{
      ..., n, cp, ...
      \xrightarrow{pop2(k)} 
     ..., n[k \mapsto pools(i(k))], cp[k \mapsto R_3],...
    }\hspace{5mm}
    
    \inferrule[pop3]{
      cp(k) = R_3 \\
      n(k) \neq \texttt{null}\\
      n(k)->taken = \texttt{true}\\
      n(k)->next \neq n(k)
    }{
      ..., n, ...
      \xrightarrow{pop3(k)} 
     ..., n[k \mapsto n(k)->next],...
    }\hspace{5mm}
    
    \inferrule[pop4]{
      cp(k) = R_3 \\
      n(k) \neq \texttt{null}\\
      n(k)->taken = \texttt{false}\\
      n(k)->ts > maxTS(k)
    }{
      ..., maxTS, cp ...
      \xrightarrow{pop4(k)} 
     ..., maxTS[k \mapsto n(k)->ts], cp[k \mapsto R_4]...
    }\hspace{5mm}
    
    \inferrule[pop5]{
      cp(k) = R_4 
    }{
      ..., youngest, cp ...
      \xrightarrow{pop5(k)} 
     ..., youngest[k \mapsto n(k)], cp[k \mapsto R_5]...
    }\hspace{5mm}
    
    \inferrule[pop6]{
      cp(k) = R_3 \\
      n(k) \neq \texttt{null}\\
      n(k)->taken = \texttt{false}\\
      n(k)->ts \leq maxTS(k)
    }{
      ..., cp, ...
      \xrightarrow{pop6(k)} 
     ...,cp[k \mapsto R_5],...
    }\hspace{5mm}
    
    \inferrule[pop7]{
      cp(k) = R_5 \\
      i(k) < \texttt{maxThreads}-1
    }{
      ...,i, cp, ...
      \xrightarrow{pop7(k)} 
     ...,i[k \mapsto i(k)+1], cp[k \mapsto R_2],...
    }\hspace{5mm}
    
    \inferrule[pop8]{
      cp(k) = R_5 \\
      youngest(k) = \texttt{null} \vee (youngest(k) \neq \texttt{null} \wedge youngest->taken)
    }{
      ...,success, cp, ...
      \xrightarrow{pop7(k)} 
     ...,success[k \mapsto \texttt{false}], cp[k \mapsto R_6],...
    }\hspace{5mm}
    
    \inferrule[com-pop]{
      cp(k) = R_5 \\
      youngest(k) \neq \texttt{null} \\
      youngest(k) = m \\
      d = m->data\\
      m->taken = false \\
      m' = m[taken \mapsto true]
    }{
      ...,success, youngest, cp, ...
      \xrightarrow{com(pop,d,k)} 
     ...,success[k\mapsto true], youngest[k \mapsto m'], cp[k \mapsto R_6],...
    }\hspace{5mm}
    
    \inferrule[pop9]{
      cp(k) = R_6 \\
      success(k) = \texttt{false}
    }{
      ..., cp, ...
      \xrightarrow{pop9(k)} 
     ..., cp[k \mapsto R_1],...
    }\hspace{5mm}
    \inferrule[ret-pop]{
      cp(k) = R_6 \\
      suc(k) = \texttt{false} \\ 
      d = yst(k)->data
    }{
      ..., cp, ...
      \xrightarrow{ret(pop,d,k)} 
     ..., cp[k \mapsto R_7],...
    }\hspace{5mm}

    \end{mathpar}
  }
 \vspace{-5mm}
  \caption{The pop derivation rules of $TSS$. We only mention the state components that are modified. Unmentioned state components have the names in the algorithm in the pre-state. $n' = n[taken \mapsto expr]$ is short for $n'->taken = expr$ and all the other fields of $n$ and $n'$ are the same.
  }
  \label{fig:transitions:TSSPop}
\vspace{-6mm}
\end{figure}

\begin{lem}
$TSS$ is a $C\cup R\cup Com(pop)$-refinement of $AbsS$. 
\end{lem}
\begin{proof}
We show that the relation $\mathit{fs}_2$ defined in Section~\ref{sec:corr_tss} is a $C\cup R\cup Com(pop)$-forward simulation  from $TSS$ to $AbsS$. For readability, we recall the definition of $\mathit{fs}_2$.

Let us make some clarifications before defining the relation. In order not to confuse nodes in TS Stack and nodes in $AbsS$, we call nodes of $AbsS$ as vertices from now on. We also define ordering relation (called traverse order) among the operations in a state of $TS$. It basically reflects the traverse order of pop operations. For two push operations $m,n \in O_a$ is state $s$ we say that $m <^{tr}_s n$ iff either $myTid(m) < myTid(n)$ or $myTid(m) = myTid(n)$ and $n_s(n)$ is reachable from $n_s(m)$ using next pointers. $\geq^{tr}$ is obtained from $<^{tr}$ in the usual way.

The relation $\mathit{fs}_2 \subseteq Q_C \rightarrow Q_{AbsS}$ contains $(s,t)$ iff the following are satisfied:
\begin{itemize}
\item[\emph{Nodes}] $k \in O_t$ iff $k$ is a push operation in $s$ ($k \in O_a$) such that either it has not inserted its node to the pool yet ( $cp_s(k) = A_i$ and $i<3$) or its node is not taken by a pop ($cp_s(k) = A_i$, $i\geq 3$ and $n_s(k)->taken = false$). 
\item[\emph{Pend/Comp}] A vertex $k \in O_t$ is pending ($\ell_t(k) = (d, \texttt{PEND})$) iff $k$ satisfies the previous condition, $x_s(k) = d$ and it is not completed in $s$ ($cp_s(k) = A_i$ and $i<6$). Similarly, this vertex is completed ($\ell_t(k) = (d, \texttt{COMP})$) iff $k$ satisfies the previous condition, $x_s(k) = d$ and it is completed in $s$ ($cp_s(k) = A_6$). Pending vertices are maximal with respect to $<_t$ i.e., if $k \in O_t$ is a pending vertex, then for all $k' \in O_t$ $k \nless_t k'$.
\item[\emph{TSOrder}] If a node has a smaller timestamp than the other node in $s$, the operations that inserted them cannot be ordered reversely in $t$. More formally, let $k, k' \in O_t$ s.t. $n_s(k)-> ts \leq n_s(k')->ts$. Then, $k' \nless_t k$.
\item[\emph{TidOrder}] Order among the nodes inserted by the same threads in $s$ must be preserved among the operations that inserted them in $t$. Let $k, k' \in O_t$ s.t. $myTid_s(k) = myTid_s(k')$ and $n_s(k)->ts < n_s(k')->ts$. Then, $k <_t k'$.
\item[Frontiers] Every maximally closed or pending vertex can be removed by a pending pop. More formally, let $k \in O_t$ such that $\ell_t(k) = (\_,\texttt{PEND})$. Then, for all pops $p$, $k \in ov_t(p)$. In the other case, let $k \in O_t$ such that $\ell_t(k) = (\_,\texttt{COMP})$ and for all other $k' \in O_t$ such that $k<_t k'$, we know $\ell_t(k') = (\_,\texttt{PEND})$. Then, for all pop operations $p$, $k \in be_t(p)$ or $k \in ov_t(p)$. 
\item[\emph{MaximalOV}] If a push $k \in O_t$ is a candidate to be removed by a pop $p$, then every other push $k'$ invoked after $k$ is a candidate to be removed by $p$ since $k$ is concurrent with $p$. More formally, let $k, k' \in O_t$ such that $k <_t k'$ and there exists a pop $p$ such that $k \in be_t(p)$ or $k \in ov_t(p)$. Then, $k' \in ov_t(p)$.
\item[\emph{MinimalBE}] If a push $k \in O_t$ has finished before the pop $p$ is invoked and yet $k$ is a candidate to be removed by $p$, other pushes completed before $k$ can not be candidates to be removed by $p$ at that state. More formally, let $k, k' \in O_t$ such that $k <_t k'$ and there exists a pop $p$ such that $k' \in be_t(p)$. Then, neither $k \in be_t(p)$ nor $k \in ov_t(b)$.
\item[\emph{ReverseFrontiers}] If all immediate followers $k' \in O_t$ of a push $k \in O_t$ are concurrent with pop $p$, then $k$ is either concurrent or maximally closed with respect to $p$. More formally, let $k \in O_t$ and for all $k' \in O_t$ such that $k \in pred_{<_t}(k')$, $k' \in ov_t(p)$, where $p$ is a pop operation. Then, $k \in ov_t(p) \cup be_t(p)$. 
\item[\emph{FixReturn}] If a pop $p$ is after its commit point action in $s$, then the $rv$ value of this operation in $t$ is fixed to $youngest_s(p)->data$. More formally, Let $p$ be the pop operation such that $cp_s(p) = R_6$ and $success_s(p) = \texttt{true}$. Then, $rv_t(p) = youngest_s(p)->data$. 
\item[\emph{TraverseBefore}] If a pop operation $p$ is currently visiting node $n$, it has non-null node $y$ as the $youngest$ and there is a non-null not taken node $m$ coming before $n$ in the traverse order with a greater timestamp than $y$, then the operation that inserts $m$ must be concurrent with $p$. More formally, assume $youngest_s(p) = y$ and $ y \neq \texttt{null}$. Let $k \in O_t$ such that $n_s(k) \neq \texttt{null}$, $n_s(k)->taken = \texttt{false}$, $n_s(k) <^{tr}_s n_s(p)$ and $n_s(k)->ts \geq y->ts$. Then, $k \in ov_t(p)$.
\item[\emph{TraverseBeforeNull}] If a pop operation $p$ is currently visiting node $n$, and its $youngest$ field is \texttt{null}, then every other node $m$ coming before $n$ in the traverse order must be concurrent with $p$. More formally , let $youngest_s(p) = \texttt{null}$ and assume there exists  an operation $k \in O_t$ such that $n_s(k) \neq \texttt{null}$, $n_s(k)->taken = \texttt{false}$ and  $n_s(k) <^{tr}_s n_s(p)$. Then, $k \in ov_t(p)$. 
\item[\emph{TraverseAfter}] If a pop operation $p$ is currently visiting node $n$ that is not null and its youngest element $m$ is not null and still not taken in state $s$, then either $m$ is a candidate to be removed by $p$ in $t$ or there exists a later node $m'$ than $n$ such that $m'$ is a candidate in $t$ and it has a bigger timestamp than n. More formally, assume that there exists $k, k' \in O_t$ such that $youngest_s(p)->taken \neq false$, $youngest_s(p) = n_s(k)$ and $n_s(k') = n_s(p)$. Then, either $k \in ov_t(p) \vee k \in be_t(p)$ or there exists $k'' \in O_t$ s.t. $n_s(k'')->ts > n_s(k) ->ts$ and $k'' \in ov_t(p) \vee k'' \in be_t(p)$ and either $k' <^{tr}_s k''$ or $n_s(p) = n_s(k'') \wedge cp_s(p) = R_j \wedge j<5$. 
\end{itemize}
Next, we will show that $\mathit{fs}_2$ is really a $C\cup R\cup Com(pop)$-forward simulation relation. Except the trivial base case, we case-split on the transition rules. We first assume $s \xrightarrow{\alpha}_{TSS} s' $ and $t \in \mathit{fs}_2[s]$. Then, we find  corresponding transition $\alpha' \in \Sigma_{AbsS}$ obeying the $C\cup R\cup Com(pop)$-forward simulation relation conditions and obtain $t'$ such that $t \xrightarrow{\alpha'}_{AbsS} st $  and $t' \in \mathit{fs}_2[s']$.

We observe that if $\alpha \in C \cup R \cup Com(pop)$, then the corresponding rule in $AbsS$ is $\alpha' = \alpha$. Otherwise, $\alpha' = \epsilon$.

Let the following describe $\alpha$: $\psi \triangleright s\xrightarrow{\alpha}_{TSS} s'$ where $\psi$ is the precondition (guard) that needs to be satisfied for enabling $\alpha$ and $\psi' \triangleright t \xrightarrow{\alpha'}_{AbsS} t'$ describe the $\alpha'$ if $\alpha' \neq \epsilon$ (equivalently $\alpha' = \alpha$).

For the cases $\alpha' = \alpha$, we first need to show $\alpha'$ is enabled in state $t$ i.e., $t$ satisfies $\psi'$. If this can not be directly obtained from the information that $s$ satisfies $\psi$ and using one or two obvious conditions on $\mathit{fs}_2$ (since $t \in \mathit{fs}_2[s]$), we show the derivation in the proof. Then, $t'$ is obtained in a unique way since $AbsS$ is deterministic on its alphabet $ \Sigma_{AbsS} = C \cup R \cup Com(pop)$. The, only other thing to show is $t' \in \mathit{fs}_2[t']$. We show this by proving that $t'$ does not violate any of the conditions of the $\mathit{fs}_2$ described above. Suppose conditions on $\mathit{fs}_2$ are of the form $\forall \overline{k}. guard_{s,t}(\overline{k}) \triangleright \phi_{s,t}(\overline{k})$ where the $\overline{k}$ is a vector of operation identifiers and $\phi$ defined on states $s$ and $t$ must hold if the guard defined on $s$ and $t$ holds. We say that a vector $\overline{k_1}$ is a new instantiation of the condtion if $\overline{k_1}$ does not satisfy the $guard_{s,t}$ while relating pre-states, but it satisfies $guard_{s',t'}$ while relating post-states.

We only explain why the new instantiations due to the difference between $s'$ and $s$ or the difference between $t'$ and $t$ do not violate the conditions. We skip the instances that we assumed while relating $s$ to $t$.

For the cases in which $\alpha' = \epsilon$, we have $t'=t$ and the only thing to show is $t \in \mathit{fs}_2[s']$. Again, we only explain why the new instantiations due to the difference between $s'$ and $s$ do not violate the conditions.

In the following, we show that $\mathit{fs}_2$ is a $C \cup R \cup Com(pop)$-forward simulation relation.
\begin{itemize}
\item[\textsc{init}] $\mathit{fs}_2[{q_0}_{TSS}] =\{{q_0}_{AbsS}\}$
\item[\textsc{call-push}] The same derivation rule of $TSS$ is applied to $t$  to obtain $t'$. The premise of the rule is satisfied by $t$ trivially in the sense explained before. The new vertex $k$ is added to the $O_t$ such that $k$ is maximal, pending and every completed vertex is ordered before $k$ in $t'$. Moreover, $k$ is overlapping with every pending pop. To see that $t' \in \mathit{fs}_2[s']$ we observe the following: \emph{Nodes} condition is preserved because $k \in O_{t'}$. Since the newly added vertex $k$ is maximal and pending in $t'$, \emph{Pend/Comp} condition is preserved. \emph{Frontiers} and \emph{MaximalOV} conditions are not violated since $k$ is added to $ov(p)$ set for every pending pop operation $p$. 
\item[\textsc{push1}] We have $t' = t$ and show $t \in \mathit{fs}_2[s']$.\emph{Nodes} and \emph{Pend/Comp} conditions are still satisfied since $k$ remains to be a pending vertex. \emph{TSOrder} is still preserved. Timestamp of $n_{s'}(k)$ is maximal and every other nodes of push operations with maximal timestamp in $s'$ are pending vertices in $t$. Hence there can be no ordering between those pushes and $k$ in $t$ that can violate \emph{TSOrder}. Moreover, $k$ is maximal in $t$ which means that it cannot be ordered before another push $k'$ of which node has a lower timestamp. \emph{TidOrder} is also satisfied. Since $k$ is ordered after every completed push in $t$ and every other push by the same thread is completed, ordering required by the \emph{TidOrder} is present.
\item[\textsc{push2}] We have $t' = t$ and show $t \in \mathit{fs}_2[s']$. \emph{Nodes} and \emph{Pend/Comp} conditions are still satisfied since $k$ remains to be a pending vertex. One can also see that the \emph{TraverseBefore} condition is preserved. Let the pop $p$ visiting node $m$ and $n_{s'}(k) <^{tr}_{s'} m$. Since $k$ and $p$ are both pending in $s$ and $t \in \mathit{fs}_2[s]$, $k \in ov_t(p)$ (by the \emph{Frontiers} condition). Hence, \emph{TraverseBefore} is preserved. 
\item[\textsc{push3}] We have $t' = t$ and show $t \in \mathit{fs}_2[s']$. We consider two cases: $n_s(k)->taken$ is \texttt{true} or it is \texttt{false}. For the former case, $k \notin O_t$. The only new instantiation we check is $k \notin O_t$ does not violate \emph{Nodes} condition while relating $s'$ to $t$. 

For the latter case, we have $k \in O_t$. \emph{Nodes} and \emph{Pend/Comp} conditions are still satisfied since $k$ remains to be a pending vertex after changing $s$ to $s'$.
\item[\textsc{push4}] We have $t' = t$ and show $t \in \mathit{fs}_2[s']$. 

We consider two cases: $n_s(k)->taken$ is \texttt{true} or it is \texttt{false}. For the former case, \emph{Nodes} condition is still satisfied since $k$ remains to be not a vertex. 

For the latter case \emph{Nodes} and \emph{Pend/Comp} conditions are still satisfied since $k$ remains to be a pending vertex. \emph{TSOrder} condition is still not violated since if $k'<_t k$, then $k'$ is a completed vertex in $s$ and $s'$. By the premise of the rule (which can be shown to hold for every operation at control point $A_4$) $i_s(k') < i_s(k)$ and consequently $n_{s'}(k')->ts < n_{s'}(k)->ts$. Since every other push by the thread of $k$ is completed, \emph{TidOrder} still continues to hold for the same reasons. \emph{TraverseAfter} condition is also preserved. Let $k'$ be the push and $p$ be the pop such that $n_s(k') = youngest_s(p)$, $n_s(k') \leq^{tr}_s n_s(k)$, $n_s(k')->ts < n_s(k)->ts$ and $k \in ov_t(p)$ or $k \in be_t(p)$. Assume $n_{s'}(k')->ts \geq N_{s'}(k)->ts$ after the action. Then, $k'$ must be a pending push both in $s$ and $s'$ by the premise of the derivation rule and $k' \in ov_t(p)$ must be true by \emph{Frontiers} condition and $t \in \mathit{fs}_2[s]$. Hence, the \emph{TraverseAfter} condition is preserved.
\item[\textsc{ret-push}] We consider two cases, $n_s(k)->taken$ is \texttt{false} or \texttt{true}. For the former case, we obtain $t'$ by applying \textsc{ret-push1} rule of $AbsS$. \emph{Nodes} and \emph{Pend/Comp} conditions are still satisfied since $k$ becomes a completed vertex in $t'$. \emph{Frontiers} condition still holds since although $k$ become a maximally closed vertex in $t'$, we have $k \in ov_{t'}(p)$ for all pending nodes $p$ (due to \emph{Frontiers} condition, $t \in \mathit{fs}_2[s]$ and $k$ was a pending operation in state $t$, $k \in ov_t(p)$). 

For the latter case, we obtain $t'$ by applying \textsc{ret-push2} rule of $AbsS$. \emph{Nodes} condition is still satisfied since $k \notin O_{t'}$. 
\item[\textsc{call-pop}] The same derivation rule of $TSS$ is applied to $t$  to obtain $t'$. \emph{Frontiers} condition holds for $p = k$ relating $s'$ to $t'$ since $k' \in ov_{t'}(k)$ for every pending vertex $k'$ and $k'' \in be_{t'}(p)$ for all completed vertex $k''$. $t'$ due to action $inv(pop,k)$ applied on $t$. \emph{MaximalOV} condition holds for $p = k$ since pending vertices are maximal in $t'$ and for any maximally closed vertex $k'$ in $t'$, if $k'$ is ordered before other vertex $k''$, then $k''$ is a pending operation by definition of being maximally closed and $k'' \in ov_t(k)$ due to the changes by \textsc{inv-pop} action on $t$. \emph{MinimalBE} condition holds while relating $s'$ to $t'$ for the pop $ p = k$ because only maximally closed vertices are in $be(k)$ and if a push $k'$ is ordered before a maximally closed push $k''$ in $t$, neither $k'' \in be_{t'}(k)$ (since $k''$ is not maximally closed) nor $k'' \in ov_{t'}$ (since $k''$ cannot be pending). \emph{ReverseFrontiers} condition holds while relating $s'$ to $t'$ for the pop $p=k$ because, if $k'' \in ov_{t'}(k)$ for all immediate successors of $k'$ in $t$, then $k'' $ are pending vertices (due to \emph{call-pop} action of $AbsS$), $k'$ is a maximally closed vertex and $k' \in be_{t'}(k)$ (due to \emph{call-pop} action of $AbsS$).
\item[\textsc{pop1}]We have $t' = t$ and $t \in \mathit{fs}_2[s']$.
\item[\textsc{pop2}]We have $t' = t$ and show $t \in \mathit{fs}_2[s']$. \emph{TraverseBefore} condition while relating $s'$ to $t$ still holds for $p=k$. Assume $youngest_{s'}(k) = y$ is a non-null node. Then, for all nodes $m$ in $s'$ such that $n_s(k) \leq^{tr}_{s'} m <^{tr}_{s'} n_{s'}(k)$ we have $m->ts < y->ts$ in $s'$ because  $n_s(k)->ts > m->ts$ (since $n_s(k)$ is added to the pool after $m$ by the same thread) and $y->ts \geq n_s(k)->ts$ in $s'$ (since either $youngest_{s'}(k) = n_s(k)$ or $youngest_{s'}(k)->ts > n_s(k)->ts$). \emph{TraverseAfter} does not have any new instatiations since the guard mentions the nodes after $n_s(k)$ while relating $s$ to $t$ whereas it mentions nodes after or including $n_{s'}(k)$ which contains the all nodes in the former case.
\item[\textsc{pop3}]We have $t' = t$ and $t \in \mathit{fs}_2[s']$.
\item[\textsc{pop4}]We have $t' = t$ and $t \in \mathit{fs}_2[s']$.
\item[\textsc{pop5}]We have $t' = t$ and show $t \in \mathit{fs}_2[s']$. 
\emph{TraverseBefore} condition while relating $s'$ to $t$ still holds for $p=k$ since $youngest_s(k)->ts < youngest_{s'}(k)->ts$ and \emph{TraverseBefore} holds while relating $s$ to $t$. 

\emph{TraverseAfter} condition also continues to hold for $p=k$. There are two possible cases: $youngest_s(k) = \texttt{null}$ or not. 

First, consider the former case. Since \emph{TraverseBeforeNull} is satisfied while relating $s$ to $t$, for every operation $k', k'' \in O_t$ such that  $k'' <^{tr}_s k'$ and $n_s(k') = youngest_{s'}(k)$ we have $k'' \in ov_t(k)$. Consider all such $k''$ such that $n_s(k'')->ts > n_s(k')->ts$. If there exists such a $k''$ such that $k' \in pred_{<_t}(k'')$, then $k' \in ov_t(k) \cup be_t(k)$ since \emph{ReverseFrontiers} condition holds relating $s$ to $t$. Otherwise, either $k'$ is maximal in $t$ or all the vertices $v$ ordered after $k'$ in $t$ we have $v >^{tr}_s k'$. Then, either $k'$ or one of these $v$ vertices must be maximal in $t$ and must be in $be_t(k) \cup ov_t(k)$ since \emph{Frontiers} condition holds (one of them is maximal in $t$) while relating $s$ to $t$. 

Second, assume there exists push operations $j, k'$ such that $n_s(j) = youngest_s(k) \neq \texttt{null}$ and $n_s(k') = n_s(k) = youngest_{s'}(k)$ . Since \emph{TraverseBefore} is satisfied while relating $s$ to $t$, if there exists a push $k'' <^{tr}_s k'$ such that $n_s(k'')$ is not taken and $n_s(k'')->ts \geq n_s(j)->ts$, then $k'' \in ov_t(k)$. Then, for all $k'' <^{tr}_s k'$ such that $n_s(k'')$ is not taken and $n_s(k'')->ts \geq n_s(k')->ts$, then $k'' \in ov_t(k)$ since $n_s(k')->ts \geq n_s(j)->ts$. If there exists such a $k''$ such that $k' \in pred_{<_t}(k'')$, then $k' \in ov_t(k) \cup be_t(k)$ since \emph{ReverseFrontiers} condition holds relating $s$ to $t$. Otherwise, either $k'$ is maximal in $t$ or all the vertices $v$ ordered after $k'$ in $t$ we have $v >^{tr}_s k'$. Then, either $k'$ or one of these $v$ vertices must be maximal in $t$ and must be in $be_t(k) \cup ov_t(k)$ since \emph{Frontiers} condition holds (one of them is maximal in $t$) while relating $s$ to $t$. 

\item[\textsc{pop6}] We have $t' = t$ and show $t \in \mathit{fs}_2[s']$. \emph{TraverseAfter} continues to hold while relating $s'$ to $t$ for $p=k$. Let $k', k'' \in O_t$ such that $youngest_s(k) = n_s(k')$, $n_s(k) = n_s(k'')$ and $k' \notin  ov_t(k) \cup be_t(k)$. Note that $k' <^{tr}_s k''$. Then, $n_s(k'')->ts < n_s(k')->ts$ since $n_s(k'')->ts < maxTS(k)$ and $maxTS(k) = n_s(k')->TS$ ($n_s(k')->ts$ cannot be \texttt{MAX\_INT} since $k'$ would be pending and $k' \in ov_t(k)$ otherwise). Hence, there exists another push $j$ such that $j >^{tr}_s$ and $j \in ov_t(k) \cup be_t(k)$. 

\item[\textsc{pop7}] We have $t' = t$ and $t \in \mathit{fs}_2[s']$.
\item[\textsc{pop8}] We have $t' = t$ and $t \in \mathit{fs}_2[s']$.
\item[\textsc{com-pop}] $t'$ is obtained by applying \textsc{com-pop1} rule of $AbsS$.
We first show that precondition of \textsc{com-pop1} rule of $AbsS$ si satisfied by $t$. If $com(pop,d,k)$ removes a  node $n$ such that there exists a push $k'$ such that $n_s(k') =n$ in $s$, then $k' \in O_t$ since it is non-null and not taken. Moreover, $k' \in ov_t(k) \cup be_t(k)$ since \emph{TraverseAfter} is preserved while relating $s$ to $t$ and all the nodes that come after $n_s(k)$ in terms of traverse order in $s$ have lower timestamp values than $n_s(k)->ts$ and $n_s(k)->ts \leq youngest_s(k)->ts$.

Next, we show that $t' \in \mathit{fs}_2[s']$. We case split on the conditions of $\mathit{fs}_2$ considering new instantiations.

\emph{Nodes} condition is still preserved after $k$ removes the node pushed by operation $k'$ in $s$ since $k' \notin O_{t'}$ anymore by due to $com(pop,d,k)$ action. 

\emph{Frontiers} condition is still preserved if $k$ removes the vertex $k'$ and makes another $k''$ maximally closed in $t$. Since all the other nodes $j$ ordered after $k''$ (except possibly $k'$) in $t$ are pending, $j \in ov_t(p)$ (due to \emph{Frontiers} condition while relating $s$ to $t$) for some pending pop $p \neq k$. Then, $k'' \in be_{t'}(p)$ by $com(pop,d,k)$ action. 

For the \emph{MinimalBE} condition, we do not have a new instance. If $k' \in be_{t'}(p)$ becomes true although $k' \notin be_t(p)$, we cannot have $k'' \in O_{t'}$ such that $k' \in pred_{<_{t'}}(k'')$ and $k'' \in be_{t'}(p)$ since $com(pop,d,k)$ does not add $k''$ to $ov(p)$ if its successor is not pending with respect to $p$.

\emph{ReverseFrontiers} condition is still preserved. If $k$ removes the vertex $k'$ and there exists an immediate predecessor $k''$ of $k'$ such that all of immediate successors of $k''$ are in $ov_{t'}(p)$, then $k'' \in ov_{t'}(p)$ due to the action $com(pop,d,k)$.

\emph{TraverseAfter} condition is still preserved after $k$ removes the node of push $k'$. Let $p \neq k$ be another pop operation such that $n_s(j) = youngest_s(p)$ for some push $j$ and $n_s(k')$ be the only node such that $n_s(k')->ts > youngest_s(p)->ts$ and $n_s(k')$ comes after $n_s(p)$ in the traverse order of $s$ and $k' \in ov_t(p) \cup be_t(p)$. Hence, there is no $k''$ such that $n_s(k'')$ comes after $n_s(p)$ in the traverse order and $j <_t k''$ except $k'$ (i). In other direction, if for all $k'' \in O_t$ such that $n_s(k'')$ comes before $n_s(p)$ in the traverse order and $n_s(k'')->ts > youngest_s(p)->ts$ , then $k'' \in ov_t(p)$ since \emph{TraverseBefore} condition holds while relating $s$ to $t$. Then, for all $k'' \in O_t$ such that $n_s(k'')$ comes before $n_s(p)$ in the traverse order of $s$ and $k'' >_t j$ implies $k'' \in ov_t(p)$ since $n_s(k'')->ts > n_s(j)->ts$ if $k'' >_t j$ (ii). Then, for all $k'' \in O_t$ such that if $k'' >_t j$, then $k'' \in ov_t(p)$ except $k'$ due to (i) and (ii). If $j \nless_t k'$, then $j \in ov_t(p) \cup be_t(p)$ since \emph{ReverseFrontiers} hold while relating $s$ to $t$ and $j \in ov_{t'} \cup be_{t'}$ after applying the action $com(pop,d,k)$. Otherwise, if $j <_t k'$, then $k \in be_{t'}$ after applying $com(pop,d,k)$.

\emph{FixReturn} condition continues to hold. If $com(pop,d,k)$ removes the node pushed by $k'$ in $s$, then $com(pop,d,k)$ removes the vertex $k'$ (assuming data independece) and $youngest_s(k')->data = \ell_t(k')_1$. Then, $youngest_{s'}(p)->data = rv_{t'}(p)$ after applying commit actions at both sides.

\item[\textsc{pop9}]We have $t' = t$ and $t \in \mathit{fs}_2[s']$.
\item[\textsc{ret-pop}] $t'$ is obtained by applying \textsc{ret-pop} rule of $AbsS$ and $t' \in \mathit{fs}_2[s']$.
\end{itemize}
\end{proof}


\begin{thebibliography}{29}
\providecommand{\natexlab}[1]{#1}
\providecommand{\url}[1]{\texttt{#1}}
\expandafter\ifx\csname urlstyle\endcsname\relax
  \providecommand{\doi}[1]{doi: #1}\else
  \providecommand{\doi}{doi: \begingroup \urlstyle{rm}\Url}\fi

\bibitem[Abadi and Lamport(1991)]{DBLP:journals/tcs/AbadiL91}
M.~Abadi and L.~Lamport.
\newblock The existence of refinement mappings.
\newblock \emph{Theor. Comput. Sci.}, 82\penalty0 (2):\penalty0 253--284, 1991.
\newblock \doi{10.1016/0304-3975(91)90224-P}.
\newblock URL \url{http://dx.doi.org/10.1016/0304-3975(91)90224-P}.

\bibitem[Abdulla et~al.(2013)Abdulla, Haziza, Hol\'{\i}k, Jonsson, and
  Rezine]{conf/tacas/AbdullaHHJR13}
P.~A. Abdulla, F.~Haziza, L.~Hol\'{\i}k, B.~Jonsson, and A.~Rezine.
\newblock An integrated specification and verification technique for highly
  concurrent data structures.
\newblock In \emph{TACAS}, pages 324--338, 2013.

\bibitem[Alur et~al.(2000)Alur, McMillan, and Peled]{journals/iandc/AlurMP00}
R.~Alur, K.~L. McMillan, and D.~Peled.
\newblock Model-checking of correctness conditions for concurrent objects.
\newblock \emph{Inf. Comput.}, 160\penalty0 (1-2):\penalty0 167--188, 2000.

\bibitem[Amit et~al.(2007)Amit, Rinetzky, Reps, Sagiv, and
  Yahav]{conf/cav/AmitRRSY07}
D.~Amit, N.~Rinetzky, T.~W. Reps, M.~Sagiv, and E.~Yahav.
\newblock Comparison under abstraction for verifying linearizability.
\newblock In \emph{CAV '07}, volume 4590 of \emph{LNCS}, pages 477--490, 2007.

\bibitem[Bouajjani et~al.(2013)Bouajjani, Emmi, Enea, and
  Hamza]{conf/esop/BouajjaniEEH13}
A.~Bouajjani, M.~Emmi, C.~Enea, and J.~Hamza.
\newblock Verifying concurrent programs against sequential specifications.
\newblock In \emph{ESOP '13}, volume 7792 of \emph{LNCS}, pages 290--309.
  Springer, 2013.

\bibitem[Bouajjani et~al.(2015{\natexlab{a}})Bouajjani, Emmi, Enea, and
  Hamza]{DBLP:conf/icalp/BouajjaniEEH15}
A.~Bouajjani, M.~Emmi, C.~Enea, and J.~Hamza.
\newblock On reducing linearizability to state reachability.
\newblock In M.~M. Halld{\'{o}}rsson, K.~Iwama, N.~Kobayashi, and B.~Speckmann,
  editors, \emph{Automata, Languages, and Programming - 42nd International
  Colloquium, {ICALP} 2015, Kyoto, Japan, July 6-10, 2015, Proceedings, Part
  {II}}, volume 9135 of \emph{Lecture Notes in Computer Science}, pages
  95--107. Springer, 2015{\natexlab{a}}.
\newblock ISBN 978-3-662-47665-9.
\newblock \doi{10.1007/978-3-662-47666-6}.
\newblock URL \url{http://dx.doi.org/10.1007/978-3-662-47666-6}.

\bibitem[Bouajjani et~al.(2015{\natexlab{b}})Bouajjani, Emmi, Enea, and
  Hamza]{DBLP:conf/popl/BouajjaniEEH15}
A.~Bouajjani, M.~Emmi, C.~Enea, and J.~Hamza.
\newblock Tractable refinement checking for concurrent objects.
\newblock In  \citet{DBLP:conf/popl/2015}, pages 651--662.
\newblock ISBN 978-1-4503-3300-9.
\newblock \doi{10.1145/2676726.2677002}.
\newblock URL \url{http://doi.acm.org/10.1145/2676726.2677002}.

\bibitem[Derrick et~al.(2011)Derrick, Schellhorn, and Wehrheim]{Derrick2011}
J.~Derrick, G.~Schellhorn, and H.~Wehrheim.
\newblock \emph{Verifying Linearisability with Potential Linearisation Points},
  pages 323--337.
\newblock Springer Berlin Heidelberg, Berlin, Heidelberg, 2011.
\newblock ISBN 978-3-642-21437-0.

\bibitem[Dodds et~al.(2015)Dodds, Haas, and Kirsch]{DBLP:conf/popl/DoddsHK15}
M.~Dodds, A.~Haas, and C.~M. Kirsch.
\newblock A scalable, correct time-stamped stack.
\newblock In  \citet{DBLP:conf/popl/2015}, pages 233--246.
\newblock ISBN 978-1-4503-3300-9.
\newblock \doi{10.1145/2676726.2676963}.
\newblock URL \url{http://doi.acm.org/10.1145/2676726.2676963}.

\bibitem[Dragoi et~al.()Dragoi, Gupta, and Henzinger]{conf/cav/DragoiGH13}
C.~Dragoi, A.~Gupta, and T.~A. Henzinger.
\newblock Automatic linearizability proofs of concurrent objects with
  cooperating updates.
\newblock In \emph{CAV '13}, volume 8044 of \emph{LNCS}, pages 174--190.
  Springer.

\bibitem[Filipovic et~al.(2010)Filipovic, O'Hearn, Rinetzky, and
  Yang]{journals/tcs/FilipovicORY10}
I.~Filipovic, P.~W. O'Hearn, N.~Rinetzky, and H.~Yang.
\newblock Abstraction for concurrent objects.
\newblock \emph{Theor. Comput. Sci.}, 411\penalty0 (51-52):\penalty0
  4379--4398, 2010.

\bibitem[Gorelik and Hendler(2013)]{DBLP:conf/podc/GorelikH13}
M.~Gorelik and D.~Hendler.
\newblock Brief announcement: an asymmetric flat-combining based queue
  algorithm.
\newblock In P.~Fatourou and G.~Taubenfeld, editors, \emph{{ACM} Symposium on
  Principles of Distributed Computing, {PODC} '13, Montreal, QC, Canada, July
  22-24, 2013}, pages 319--321. {ACM}, 2013.
\newblock ISBN 978-1-4503-2065-8.
\newblock \doi{10.1145/2484239.2484279}.
\newblock URL \url{http://doi.acm.org/10.1145/2484239.2484279}.

\bibitem[Hamza(2015)]{DBLP:conf/netys/Hamza15}
J.~Hamza.
\newblock On the complexity of linearizability.
\newblock In A.~Bouajjani and H.~Fauconnier, editors, \emph{Networked Systems -
  Third International Conference, {NETYS} 2015, Agadir, Morocco, May 13-15,
  2015, Revised Selected Papers}, volume 9466 of \emph{Lecture Notes in
  Computer Science}, pages 308--321. Springer, 2015.
\newblock ISBN 978-3-319-26849-1.
\newblock \doi{10.1007/978-3-319-26850-7}.
\newblock URL \url{http://dx.doi.org/10.1007/978-3-319-26850-7}.

\bibitem[Hendler et~al.()Hendler, Shavit, and
  Yerushalmi]{conf/spaa/HendlerSY04}
D.~Hendler, N.~Shavit, and L.~Yerushalmi.
\newblock A scalable lock-free stack algorithm.
\newblock In \emph{SPAA 2004}, pages 206--215. ACM.

\bibitem[Hendler et~al.(2010)Hendler, Incze, Shavit, and
  Tzafrir]{DBLP:conf/spaa/HendlerIST10}
D.~Hendler, I.~Incze, N.~Shavit, and M.~Tzafrir.
\newblock Flat combining and the synchronization-parallelism tradeoff.
\newblock In F.~M. auf~der Heide and C.~A. Phillips, editors, \emph{{SPAA}
  2010: Proceedings of the 22nd Annual {ACM} Symposium on Parallelism in
  Algorithms and Architectures, Thira, Santorini, Greece, June 13-15, 2010},
  pages 355--364. {ACM}, 2010.
\newblock ISBN 978-1-4503-0079-7.
\newblock \doi{10.1145/1810479.1810540}.
\newblock URL \url{http://doi.acm.org/10.1145/1810479.1810540}.

\bibitem[Henzinger et~al.(2013)Henzinger, Sezgin, and
  Vafeiadis]{conf/concur/HenzingerSV13}
T.~A. Henzinger, A.~Sezgin, and V.~Vafeiadis.
\newblock Aspect-oriented linearizability proofs.
\newblock In \emph{CONCUR}, pages 242--256, 2013.

\bibitem[Herlihy and Wing(1990)]{journals/toplas/HerlihyW90}
M.~Herlihy and J.~M. Wing.
\newblock Linearizability: A correctness condition for concurrent objects.
\newblock \emph{ACM Trans. Program. Lang. Syst.}, 12\penalty0 (3):\penalty0
  463--492, 1990.

\bibitem[Hoffman et~al.(2007)Hoffman, Shalev, and
  Shavit]{DBLP:conf/opodis/HoffmanSS07}
M.~Hoffman, O.~Shalev, and N.~Shavit.
\newblock The baskets queue.
\newblock In E.~Tovar, P.~Tsigas, and H.~Fouchal, editors, \emph{Principles of
  Distributed Systems, 11th International Conference, {OPODIS} 2007,
  Guadeloupe, French West Indies, December 17-20, 2007. Proceedings}, volume
  4878 of \emph{Lecture Notes in Computer Science}, pages 401--414. Springer,
  2007.
\newblock ISBN 978-3-540-77095-4.

\bibitem[Liang and Feng(2013)]{DBLP:conf/pldi/LiangF13}
H.~Liang and X.~Feng.
\newblock Modular verification of linearizability with non-fixed linearization
  points.
\newblock In H.~Boehm and C.~Flanagan, editors, \emph{{ACM} {SIGPLAN}
  Conference on Programming Language Design and Implementation, {PLDI} '13,
  Seattle, WA, USA, June 16-19, 2013}, pages 459--470. {ACM}, 2013.
\newblock ISBN 978-1-4503-2014-6.
\newblock \doi{10.1145/2462156.2462189}.
\newblock URL \url{http://doi.acm.org/10.1145/2462156.2462189}.

\bibitem[Lynch and Vaandrager(1995)]{DBLP:journals/iandc/LynchV95}
N.~A. Lynch and F.~W. Vaandrager.
\newblock Forward and backward simulations: I. untimed systems.
\newblock \emph{Inf. Comput.}, 121\penalty0 (2):\penalty0 214--233, 1995.
\newblock \doi{10.1006/inco.1995.1134}.
\newblock URL \url{http://dx.doi.org/10.1006/inco.1995.1134}.

\bibitem[Morrison and Afek(2013)]{DBLP:conf/ppopp/MorrisonA13}
A.~Morrison and Y.~Afek.
\newblock Fast concurrent queues for x86 processors.
\newblock In A.~Nicolau, X.~Shen, S.~P. Amarasinghe, and R.~W. Vuduc, editors,
  \emph{{ACM} {SIGPLAN} Symposium on Principles and Practice of Parallel
  Programming, PPoPP '13, Shenzhen, China, February 23-27, 2013}, pages
  103--112. {ACM}, 2013.
\newblock ISBN 978-1-4503-1922-5.
\newblock \doi{10.1145/2442516.2442527}.
\newblock URL \url{http://doi.acm.org/10.1145/2442516.2442527}.

\bibitem[O'Hearn et~al.()O'Hearn, Rinetzky, Vechev, Yahav, and
  Yorsh]{conf/podc/OHearnRVYY10}
P.~W. O'Hearn, N.~Rinetzky, M.~T. Vechev, E.~Yahav, and G.~Yorsh.
\newblock Verifying linearizability with hindsight.
\newblock In \emph{PODC '10}, pages 85--94. ACM.

\bibitem[Rajamani and Walker(2015)]{DBLP:conf/popl/2015}
S.~K. Rajamani and D.~Walker, editors.
\newblock \emph{Proceedings of the 42nd Annual {ACM} {SIGPLAN-SIGACT} Symposium
  on Principles of Programming Languages, {POPL} 2015, Mumbai, India, January
  15-17, 2015}, 2015. {ACM}.
\newblock ISBN 978-1-4503-3300-9.
\newblock URL \url{http://dl.acm.org/citation.cfm?id=2676726}.

\bibitem[Schellhorn et~al.(2012)Schellhorn, Wehrheim, and
  Derrick]{DBLP:conf/cav/SchellhornWD12}
G.~Schellhorn, H.~Wehrheim, and J.~Derrick.
\newblock How to prove algorithms linearisable.
\newblock In P.~Madhusudan and S.~A. Seshia, editors, \emph{Computer Aided
  Verification - 24th International Conference, {CAV} 2012, Berkeley, CA, USA,
  July 7-13, 2012 Proceedings}, volume 7358 of \emph{Lecture Notes in Computer
  Science}, pages 243--259. Springer, 2012.
\newblock ISBN 978-3-642-31423-0.
\newblock \doi{10.1007/978-3-642-31424-7}.
\newblock URL \url{http://dx.doi.org/10.1007/978-3-642-31424-7}.

\bibitem[Vafeiadis()]{conf/cav/Vafeiadis10}
V.~Vafeiadis.
\newblock Automatically proving linearizability.
\newblock In \emph{CAV '10}, volume 6174 of \emph{LNCS}, pages 450--464.

\bibitem[Vafeiadis(2008)]{phd/Vafeiadis08}
V.~Vafeiadis.
\newblock \emph{Modular fine-grained concurrency verification}.
\newblock PhD thesis, University of Cambridge, 2008.

\bibitem[Vafeiadis(2009)]{conf/vmcai/Vafeiadis09}
V.~Vafeiadis.
\newblock Shape-value abstraction for verifying linearizability.
\newblock In \emph{VMCAI '09: Proc. 10th Intl. Conf. on Verification, Model
  Checking, and Abstract Interpretation}, volume 5403 of \emph{LNCS}, pages
  335--348. Springer, 2009.

\bibitem[Vafeiadis et~al.()Vafeiadis, Herlihy, Hoare, and
  Shapiro]{conf/ppopp/VafeiadisHHS06}
V.~Vafeiadis, M.~Herlihy, T.~Hoare, and M.~Shapiro.
\newblock Proving correctness of highly-concurrent linearisable objects.
\newblock In \emph{PPOPP '06}, pages 129--136. ACM.

\bibitem[Zhu et~al.(2015)Zhu, Petri, and Jagannathan]{DBLP:conf/cav/ZhuPJ15}
H.~Zhu, G.~Petri, and S.~Jagannathan.
\newblock Poling: {SMT} aided linearizability proofs.
\newblock In D.~Kroening and C.~S. Pasareanu, editors, \emph{Computer Aided
  Verification - 27th International Conference, {CAV} 2015, San Francisco, CA,
  USA, July 18-24, 2015, Proceedings, Part {II}}, volume 9207 of \emph{Lecture
  Notes in Computer Science}, pages 3--19. Springer, 2015.
\newblock ISBN 978-3-319-21667-6.
\newblock \doi{10.1007/978-3-319-21668-3}.
\newblock URL \url{http://dx.doi.org/10.1007/978-3-319-21668-3}.

\end{thebibliography}
\end{document}